\documentclass[11p,reqno]{amsart}

\usepackage[T1]{fontenc}
\usepackage{graphicx}
\usepackage{amssymb,amsthm,amsmath,mathrsfs,bm,braket,marginnote}
\usepackage{enumerate}
\usepackage{appendix}
\usepackage[colorlinks=true, pdfstartview=FitV, linkcolor=blue, citecolor=blue, urlcolor=blue]{hyperref}
\usepackage{multirow}

\usepackage{pgf}
\usepackage{pgfplots}
\usepackage{tikz}
\usetikzlibrary{arrows,calc}
\usepackage{verbatim}
\usetikzlibrary{decorations.pathreplacing,decorations.pathmorphing}
\usepackage[numbers,sort&compress]{natbib}
\usepackage{dsfont}

\numberwithin{equation}{section}
\newtheorem{theorem}{Theorem}[section]

\newtheorem{coro}[theorem]{Corollary}

\newtheorem{proposition}[theorem]{Proposition}

\theoremstyle{remark}
\newtheorem{remark}[theorem]{Remark}

 \reversemarginpar

\begin{document}

\title[]{Dip-ramp-plateau for  Dyson Brownian motion from the identity on $U(N)$ }

\date{}

\author{Peter J. Forrester}
\address{School of Mathematical and Statistics, The University of Melbourne, Victoria 3010, Australia}
\email{pjforr@unimelb.edu.au}

\author{Mario Kieburg}
\address{School of Mathematical and Statistics, The University of Melbourne, Victoria 3010, Australia}
\email{m.kieburg@unimelb.edu.au}

\author{Shi-Hao Li}
\address{Department of Mathematics, Sichuan University, Chengdu, 610064, China}
\email{lishihao@lsec.cc.ac.cn}

\author{Jiyuan Zhang}
\address{School of Mathematical and Statistics, The University of Melbourne, Victoria 3010, Australia}
\email{jiyuanzhang.ms@gmail.com}

\dedicatory{}


\begin{abstract}
In a recent work the present authors have shown that the eigenvalue probability density function for Dyson Brownian motion from the identity on $U(N)$ is an example of a newly identified class
of random unitary matrices called cyclic P\'olya ensembles. In general the latter exhibit a structured form of the correlation kernel. Specialising to the case of
Dyson Brownian motion from the identity on $U(N)$ allows the moments of the spectral density, and the spectral form factor $S_N(k;t)$, to be evaluated explicitly in terms
of a certain hypergeometric polynomial. Upon transformation, this can be identified in terms of a Jacobi polynomial with parameters $(N(\mu - 1),1)$,
where $\mu = k/N$ and $k$ is the integer labelling  the Fourier coefficients. From existing results in the literature for the asymptotics of the latter,
the asymptotic forms of the moments of the spectral density can be specified, as can $\lim_{N \to \infty} {1 \over N} S_N(k;t) |_{\mu = k/N}$. These in turn
allow us to give a quantitative description of the large $N$ behaviour of the average $ \langle  | \sum_{l=1}^N e^{ i k x_l}  |^2  \rangle$. The latter exhibits
a dip-ramp-plateau effect, which is attracting recent interest from the viewpoints of many body quantum chaos, and the scrambling of information in black holes.
\end{abstract}

\maketitle

\section{Introduction}
\subsection{Focus of the paper}\label{S1.1}
Dyson Brownian motion on the group $U(N)$ of $N \times N$ complex unitary matrices is of long standing interest in random matrix
theory. This process is contained in, and gets its name from, Dyson's 1962 paper \cite{Dy62b}. Actually Dyson's paper is better known for
the construction of a Brownian motion model for the eigenvalues of particular spaces of Hermitian Gaussian random matrices; see
the monographs \cite{Ka16,EY17} for developments in relation to this topic. With $U(N)$ being a Lie group, the notion of a corresponding
Brownian motion and its theoretical development can be found in earlier works of It\^{o} \cite{It50}, Yoshida \cite{Yo52} and Hunt \cite{Hu56};
we owe our knowledge of these references to the Introduction given in \cite{AKMV10}. While Dyson's objective was a Brownian motion
theory of the eigenvalues, the objective of the earlier works related to the group elements $U_t$ say; recent works which
relate to this latter aspect include \cite{Bi97, Ra97, Ha01, Ke17}.
The matrix evolution can be approximated  numerically by the specification
\begin{equation}\label{1.0}
U_{t + \delta t}(N) = U_t e^{ i \sqrt{\delta t} M} , \qquad 0 < \delta t \ll 1,
\end{equation}
where $M$ is an element of the Gaussian unitary ensemble of $N \times N$ complex Hermitian matrices; see e.g.~\cite{PS91} or
\cite[Eq.~(11.26)]{Fo10}.
Wolfram Mathematica, with the support of random matrix theory added in version 11 or later, can easily implement (\ref{1.0}) in a few lines of code \cite{Wo11}.
 This gives a matrix Brownian motion path on $U(N)$ sampled at discrete time intervals $\delta t$. This path is simple to display in terms of the corresponding eigenvalues; see Figure \ref{Dfig1} for an example.

In this paper our aim is to begin with Dyson's specification
of the Brownian motion theory of the eigenvalues on $U(N)$, to specialise to the case that the initial condition is the identity, and to
quantify specific statistical properties of the corresponding dynamical state, most notably the dip-ramp-plateau effect; see Section \ref{S1.3} in relation to this.
Particular attention will be payed to the Fourier components of the eigenvalue density at a specific time, and to
the spectral form factor (structure function). This line of study is enabled by the recent identification \cite{KLZF20} of a cyclic P\'olya
ensemble structure in relation to the eigenvalue  probability density function (PDF) for this model, supplemented by transformation and
asymptotic theory associated with certain hypergeometric polynomials.  Insights from other sources have seen progress
on many fronts in the study of the eigenvalues for Dyson Brownian motion --- interpreted broadly --- in recent years; references include \cite{De10,GS15,We16,As19,HL19,LSY19, BCL21}.

\begin{figure*}
\centering
\includegraphics[width=0.75\textwidth]{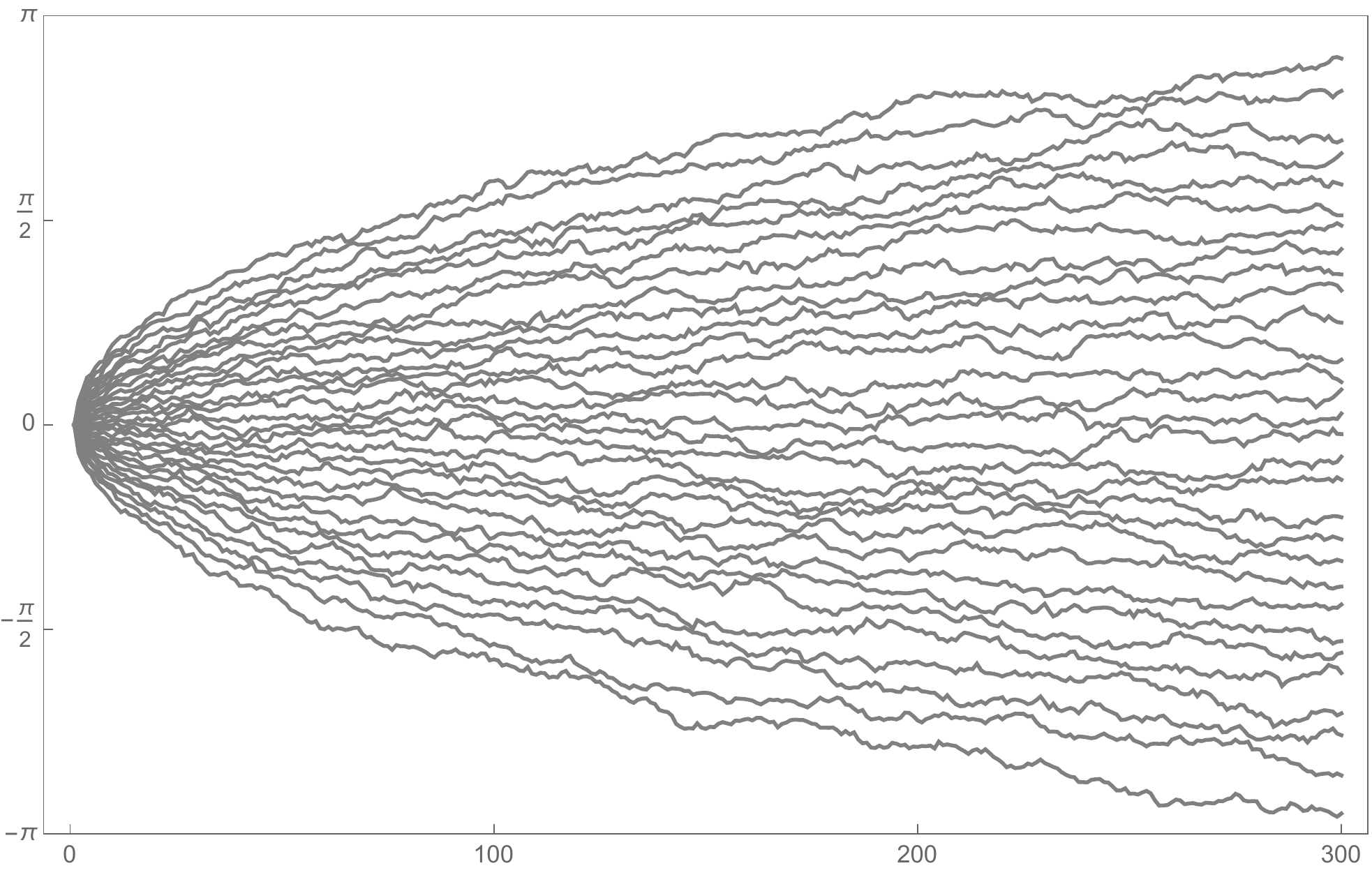}
\caption{Plot of the eigenvalues of a sequence of $m = 300$ unitary matrices with $N=30$ exhibiting Dyson Brownian motion,
generated according to the approximation (\ref{1.0}) with $\sqrt{\delta t} = 0.02$. This corresponds to a scaled time $t$, defined below (\ref{3.0c2}),
of $t=m N(\delta t) = 3.6$. Theory predicts the eigenvalue support to approach $(-\pi, \pi)$ as $t \to 4^-$.}
\label{Dfig1}
\end{figure*}

\subsection{Definition of key statistical quantities}\label{S1.2}

Unitary matrices have eigenvalues on the unit circle and so can be parametrised as $\{e^{ i  x_j(t)} \}_{j=1}^N$, with
$- \pi < x_j(t) \le \pi$ angles.  Let $\langle \cdot \rangle$ denote an ensemble average with respect to the joint PDF for these
angles. The eigenvalue density $\rho_{(1),N}(x;t)$ and the two-point
correlation function $\rho_{(2),N}(x,y;t)$ are specified as such ensemble averages of sums of Dirac delta functions
\begin{equation}\label{A.1a}
\rho_{(1),N}(x;t)  = \Big \langle \sum_{j=1}^N \delta (x - x_j(t)) \Big \rangle, \qquad
\rho_{(2),N}(x,y;t)  = \Big \langle \sum_{p,q=1 \atop p \ne q}^N \delta (x - x_p(t))  \delta (y - x_q(t)) \Big \rangle.
\end{equation}
The Fourier components of the spectral density are, up to proportionality, given by
\begin{equation}\label{A.2a}
\int_{-\pi}^{\pi} e^{-  i k x} \rho_{(1),N}(x;t) \, \mathrm dx = 
\Big \langle \sum_{l=1}^N e^{- i k x_l(t)} \Big \rangle.
\end{equation}
The spectral form factor is specified by
\begin{align}\label{S.4}
S_N(k;t)  := {\rm Cov} \, \Big ( \sum_{l=1}^N e^{ i k x_l(t) },  \sum_{l=1}^N e^{- i k x_l(t)} \Big ) 
 := \Big \langle   \sum_{l,l'=1}^N e^{  i k (x_l(t) - x_{l'}(t))} \Big \rangle - \Big \langle   \sum_{l=1}^N e^{i k x_l(t) }  \Big \rangle    \Big \langle  \sum_{l=1}^N e^{- i k x_l(t)}  \Big \rangle.
\end{align}

Denote the truncated
 two-point correlation (also referred to as the two-point cluster function \cite{Me04}) by
   \begin{equation}\label{S.1} 
\rho_{(2),N}^T(x_1,x_2;t) := \rho_{(2),N}(x_1,x_2;t) - \rho_{(1),N}(x_1;t) \rho_{(1),N}(x_2;t).
\end{equation}
A straightforward calculation (see e.g.~\cite[proof of Prop.~2.1]{Fo22c})
shows that the spectral form factor can be written
 in terms of the truncated two-point correlation and the density according to
\begin{align}\label{S.3}
S_N(k;t) & = \int_{- \pi }^{ \pi } \mathrm dx  \int_{- \pi }^{ \pi } \mathrm dy \, e^{ i k (x - y)}  \Big ( \rho_{(2),N}^T(x,y;t) + \rho_{(1),N}(x;t) \delta(x - y) \Big ) \nonumber \\
& = N + \int_{- \pi}^{\pi} \mathrm dx  \int_{-\pi}^{\pi} \mathrm dy \, e^{ i k (x - y)}   \rho_{(2),N}^T(x,y;t),
\end{align}
where underlying the second line is the requirement that the density integrated over its support equals $N$.

The Fourier components of the spectral density, and the spectral form factor, show themselves in the consideration of the ensemble average of the quantity
\begin{equation}\label{S.5m}
  \Big |  \sum_{l=1}^N e^{ i k x_l(t) } \Big |^2 , \qquad k \in \mathbb Z.
\end{equation}
Thus from (\ref{S.4}) one has the decomposition
\begin{equation}\label{S.5}
\bigg \langle  \Big |  \sum_{l=1}^N e^{ i k x_l(t) } \Big |^2 \bigg  \rangle =  S_N(k;t) +   \bigg |  \bigg \langle  \sum_{l=1}^N e^{ i k x_l(t)} \bigg  \rangle   \bigg |^2.
\end{equation}
 For $\{ x_l \}$ viewed as scaled energy levels,
both the spectral form factor and the average (\ref{S.5})
 were put forward as  probes of quantum chaos in the early literature on the subject
\cite{Be85,LLJP86}. In particular, 
in relation to the average (\ref{S.5m}), the work \cite{LLJP86} 
drew attention to 
an effect referred to as a correlation hole in the corresponding graphical shape  as
a function of $k$, as distinguished from its absence for integrable spectra. The essential point is that
for a chaotic system the spectral form factor goes to zero for $k$ small enough  and then increases 
before saturating (measured on some
$N$ dependent scale), while  the second term on the RHS of (\ref{S.5}) continues to decrease on this
same $N$ dependent scale, with its initial decay as a function of $N$ and $k$ characterised by a distinct $N$ dependent scale.
Under the name dip-ramp-plateau,
or more accurately slope-dip-ramp-plateau \cite{Sh16,Ya20,CES21},
this effect has received renewed attention as a probe of many body quantum chaos
 \cite{CHLY17,TGS18,CL18,CMC19,CH19,LPC21}, and also in relation to the scrambling of information in black
 holes \cite{C+17,CMS17}. This in turn has prompted a revival of interest in the consideration of analytic
 properties of model systems for which (\ref{S.5m}) exhibits a dip-ramp-plateau shape
  \cite{BH97,Ok19,Fo21a,Fo21b,MH21,CES21,VG21,VG22}.

  In the present work, we identify Dyson Brownian
 motion on $U(N)$ from the identity as one of the rare known solvable cases of the dip-ramp-plateau effect, joining the 
 Gaussian unitary ensemble (GUE) \cite{BH97} and Laguerre unitary ensemble (LUE) \cite{Fo22b}. Moreover, our exact solution
 exhibits analytic properties not seen in any of the previous solvable cases.
To appreciate this point, general features of this effect, albeit based on non-rigorous analysis and illustrated to the extent possible
by GUE and LUE, should first be revised.

\subsection{Heuristics of (slope)-dip-ramp-plateau}\label{S1.3}
 The dip-ramp-plateau effect occurs for large $N$ and relates to several regimes as $k$ varies. These are as
$k$ varies from order unity to of order proportional to $N$ (this relates to the slope, dip and ramp), and as 
$k$ varies to be proportional to $N$ (this relates to the ramp and plateau) giving rise to a well defined limiting quantity in
$\mu$.  Specifically, in relation to the latter we write $\mu = k/N$, and assume $k > 0$. The significance of this scaling is that we have $k x_j(t) = \mu X_j(t)$, where 
 $X_j(t) = N x_j(t)$. Since the spacing between the angles $\{ x_j(t) \}$ in the bulk region is ${\rm O}(1/N)$, the spacing between
 $\{X_j(t)\}$ in the bulk is $O(1)$. Hence in the variable $\mu$ one would expect a well defined $N \to \infty$ limit of the
 spectral form factor, corresponding to a deformation of the Fourier transform of the bulk scaled limiting form of the term
 in the big brackets of the integrand of the first line of (\ref{S.3}). The deformation is
 due to the variation of  the spacing between $\{X_j(t)\}$ away from their bulk 
 scaled value across all of the spectrum (in particular near the edges). 
In the ensuing discussion the fact that $k$ is also an integer plays no role, so we can take $\mu$ as continuous.

For large $N$ and an appropriate scaling of time, the eigenvalue density for Dyson Brownian
 motion on $U(N)$ from the identity has the large $N$ form \cite{Bi97}
\begin{equation}\label{Mam}
\rho_{(1),N}(x;t) \sim {N \over 2 \pi } 
\rho_{(1),\infty}(x;t),
\end{equation}
where  $\rho_{(1),\infty}(x;t)$ is independent of $N$.  It follows that the second term in (\ref{S.5}) can be approximately expressed  for large $N$ by
\begin{equation}\label{Ma1}
\Big ( {N \over 2 \pi} \Big )^2  \bigg |   \int_{- \pi }^{ \pi } \rho_{(1),\infty}(x;t) \, e^{ i \mu N x } \, \mathrm dx \bigg |^2.
\end{equation}
For $0 < t < 4$ it is known \cite{Bi97} (see also \S \ref{S3.1}) that the support of $\rho_{(1),\infty}(x;t)$ is $(-L_0(t),L_0(t))$ for some $0 < L_0(t) <  \pi $ and that $\rho_{(1),\infty}(x;t)$ goes to zero as a square root at the end points $\pm L_0(t)$ with some amplitude $A(t)$.
 According to Fourier transform theory \cite{Li58}, the integral  in (\ref{Ma1})  then displays a leading order decay  
 \begin{equation}\label{dy}
   \sqrt{ \pi }  A(t)  {\cos ( \mu N L_0(t) - 3 \pi/4) \over (\mu N)^{3/2}}.
\end{equation} 
Hence (\ref{Ma1}) itself has the decay for $N \mu$ large
\begin{equation}\label{Ma2}
  { (A(t))^2  \over 4 \pi} {\cos^2 ( \mu N L_0(t) -3 \pi/ 4 )\over N \mu^3}.
\end{equation}
On the other hand, for $t > 4$ the support of $\rho_{(1),\infty}(x;t)$ is all of $(-\pi,\pi]$ and it is a periodic analytic function. As such the integral in (\ref{Ma1}) decays exponentially fast in $N$, and so in this regime (\ref{Ma2}) is replaced by
\begin{equation}\label{Ma2b}
N^2 e^{-c(t) \mu N}
\end{equation}
for some $c(t)>0$.
Since (\ref{Ma1}) takes on the value $N^2$ for $\mu = 0$ and then decreases according to (\ref{Ma2}) or (\ref{Ma2b}), its decay as $\mu$ increases is responsible for the initial slope
 in the terminology slope-dip-ramp-plateau. A significant qualitative  feature is that the functional form of the slope changes from being algebraic for $ t < 4$ to decaying exponentially for
 $t > 4$.

 The circumstance of a square root singularity at the boundary of the eigenvalue support is a feature of both the GUE and LUE. In both cases the second term in (\ref{S.5}) permits
 a special function evaluation from which the analogue of (\ref{dy}) is readily verified. Considering the GUE for definiteness, one has \cite{Ul85}
 \begin{equation}\label{X.1}
 \Big \langle \sum_{j=1}^N e^{i k \lambda_j}  \Big \rangle_{\rm GUE} = e^{-k^2/4} L_{N-1}^{(1)}(k^2/2),
 \end{equation}
 where $L_n^{(a)}(x)$ denotes the Laguerre polynomial. To proceed further the analogue of the scaled wavenumber $\mu$ is required. This is the variable $\tau_b$ specified
 by $k = 2 \sqrt{2N} \tau_b$. Analogous to the discussion at the beginning of this subsection, the justification is that then that $k \lambda_j = \tau_b X_j$, where $X_j = 2 \sqrt{2N} \lambda_j$ has the property that in the bulk the eigenvalue spacings are of order unity; see e.g.~\cite[\S 3.3]{Fo21a}. Plancherel-Rotach asymptotics of the Laguerre polynomial can
 be used to exhibit that for $N \tau_b$ large but $\tau_b$ small \cite[Eqns.~(3.41) and (3.39)]{Fo21a} (see also \cite[working at end of \S 2]{GM95}),
  \begin{equation}\label{dyG}
   \Big \langle \sum_{j=1}^N e^{i 2 \sqrt{2N} \tau_b \lambda_j}  \Big \rangle_{\rm GUE}  \sim {1  \over 2 \sqrt{2 \pi N} \tau_b^{3/2}} \cos (4 N \tau_b - 3 \pi / 4).
  \end{equation}

A heuristic understanding of the large $N$ form of spectral form factor for $\mu$ fixed is based on Dyson Brownian
 motion on $U(N)$  being a member of the class of random matrix models which have a unitary symmetry.  Thus the matrix distribution at a given time is unchanged by conjugation with a fixed $U \in U(N)$. A universal feature of such matrix ensembles is the functional form of the bulk truncated two-point correlation function (see e.g.~\cite[rewrite of (7.2)]{Fo10})
\begin{equation}\label{Ma3}
\rho_{(2)}^{T \, \rm bulk}(x,y) = - {\sin [ \pi \rho (x - y) ])^2 \over ( \pi (x - y))^2}.
\end{equation}
The quantity $\rho$ is the local eigenvalue density, the precise value of which depends on the choice of units underlying the bulk scaling, which in broad terms requires a centring of the coordinates away from the boundary of support, and a rescaling so that the eigenvalue density is of order unity.
According to (\ref{S.3}) the spectral form factor is determined by $\rho_{(2)}^{T}$. However the average must be taken over the entire spectrum simultaneous to the scaling
of $k$, or equivalently the fixing of $\rho$ in (\ref{Ma3}) (on this last point recall the discussion of the first paragraph of this subsection). This average was first computed in a form
suitable for taking the large $N$ scaled limit of the corresponding spectral form factor
$S_N^{(\rm G)}(k)$ for the GUE by Br\'ezin and Hikami \cite{BH97}, with the result
\begin{equation}\label{Ma3.1}
\lim_{N \to \infty} {1 \over N} S_N^{(\rm G)}( 2 \sqrt{2N} \tau_b ) = 
  \left \{ \begin{array}{ll} {2 \over \pi} (\tau_b \sqrt{(1 - \tau_b^2)} + {\rm Arcsin} \, \tau_b ), & 0 < \tau_b  < 1, \\
  1, & \tau_b > 1.
  \end{array} \right.
   \end{equation}
Moreover it was shown in \cite{BH97} that a heuristic analysis based on (\ref{Ma3}) could reproduce the exact result (\ref{Ma3.1}).
Later this same argument was shown to also reproduce a newly obtained exact evaluation of 
the analogous  limit of the spectral form factor $S_N^{(\rm L)}((k)$ for the LUE  with Laguerre parameter $a$ fixed \cite[Eq.~(1.26)]{Fo22b},
\begin{equation}\label{Ma3.2}
\lim_{N \to \infty} {1 \over N} S_N^{(\rm L)}(k) = {\rm Arctan} \, k
  \end{equation} 
(here no simultaneous scaling of $k$ is required since a feature of the LUE is that for large $N$ the average spacing between
eigenvalues in the bulk is of order unity).

The basic idea applied in the present context is to replace the bulk density $\rho$ in (\ref{Ma3}) by the asymptotic global density $N \rho_{(1), \infty}((x+y)/2;t)/(2 \pi)$; recall (\ref{Mam}). Substituting in the second expression of (\ref{S.3}) and changing variables $w=N (x-y)$, $u=(x+y)/2$ then gives the large $N$ prediction
\begin{equation}\label{Ma4}
{1 \over N} S_N(\mu N;t) \sim 1 - {1 \over \pi^2} \int_{-\infty}^\infty \mathrm dw \int_{-\pi }^{\pi } \mathrm du \,
{\sin^2(  \rho_{(1),\infty}(u;t )w / 2) \over w^2}
e^{ i w \mu}.
\end{equation}
Changing the order of integration, the integral over $w$ can be computed explicitly, showing that for $\mu > 0$
\begin{equation}\label{Ma5}
{1 \over N} S_N(\mu N;t) \sim 1 - {1 \over \pi}  \int_0^{u^*}  (\rho_{(1),\infty}(u;t) -  \mu   ) \, \mathrm du.
\end{equation}
Here use has been made of the fact that $\rho_{(1),\infty}(u;t)$ is even in $u$, and $u^* = u^*(\mu)$ is such that
\begin{equation}\label{Ma6}
  \rho_{(1),\infty}(u^*;t) = \mu.
\end{equation}
If (\ref{Ma6}) has no solution with the LHS always larger than the RHS, 
(\ref{Ma4}) implies that $u^*$ in (\ref{Ma5}) is to be replaced by $\pi$. Then the first integral in (\ref{Ma5}) evaluates to $\pi$, so implying
\begin{equation}\label{Ma7}
{1 \over N} S_N(\mu N;t) \sim  \mu.
\end{equation}
If there is no solution with the LHS always smaller than the RHS then (\ref{Ma4}) implies that the integral in (\ref{Ma5}) is absent, and so
\begin{equation}\label{Ma8}
{1 \over N} S_N(\mu N;t) \sim 1.
\end{equation}
The behaviour (\ref{Ma7}) is referred to as a ramp, and (\ref{Ma8}) as a plateau, in the terminology dip-ramp-plateau.

Equating the RHS of (\ref{Ma7}) multiplied by $N$ with (\ref{Ma2}) gives $\mu \asymp N^{-1/2}$ or equivalently $k \asymp N^{1/2}$. The latter is referred
to as the dip-"time". It quantifies the order in $N$ of the minimum of the graph of (\ref{S.5m}).
We put time within quotation marks here as the terminology is confusing in the present setting of Dyson Brownian motion on $U(N)$,
with the term time already being used in the context of the evolution. Instead we will refer to this as the dip-wavenumber.
Its calculated value, which is valid for $0 < t < 4$ is the same as for the GUE; see e.g.~\cite[second last paragraph of \S 1.1]{CES21}. 
This follows by equating (\ref{dyG}) with (\ref{Ma3.1}) turned into a large $N$ statement by replacing the equals by asymptotically equals, and
multiplying by $N$.
On the other hand, as
is relevant for $t > 4$, equating the RHS of (\ref{Ma7}) multiplied by $N$ with (\ref{Ma2b}) gives $N e^{-c(t) \mu N} = \mu$. This equation relates to the Lambert
$W$-function and implies $k  \asymp \log N$ for the dip-wavenumber.

In quantitive terms, it follows from the predictions of the paragraph two above that there will be a deviation from the ramp and plateau functional forms whenever the global density is not a constant. However, if the global density is strictly positive, then the prediction is that the ramp will be unaltered in the range $0 < \mu < \mu_r$ for some $\mu_r$. If the global density is bounded from above, this argument gives instead that the plateau will be unaltered in the range $\mu_p < \mu < \infty$ for some $\mu_p$. Aspects of these predictions have previously been confirmed from the exact calculation of the (scaled) limit of $S_N(k)$ for the GUE as seen by (\ref{Ma3.1})
 and for the LUE as seen by (\ref{Ma3.2}). For the GUE the global density is given by the Wigner semi-circle law, which vanishes at the endpoints. Hence we can anticipate that the ramp will always be deformed. On the other hand the Wigner semi-circle is bounded, so it is predicted that the plateau is unaltered beyond a critical value $\mu_p$. Both these quantitative features are indeed seen in the exact functional form (\ref{Ma3.1}). In the case of the LUE with Laguerre parameter $a$ fixed, the global density follows a particular
Mar\v{c}enko-Pastur functional form proportional to $x^{-1/2}(1-x)^{1/2} \chi_{0 < x < 1}$, which goes to infinity at the origin, before decreasing monotonically to zero at the right hand endpoint of the support. In this setting the heuristic working predicts that the ramp and plateau are both always deformed, which is indeed a feature of the exact solution (\ref{Ma3.2}).

As commented, Dyson Brownian
 motion on $U(N)$ from the identity has the feature that for $t > 4$, the global density is nonzero. This is in distinction to the global density for both the GUE and LUE. The significance of this feature in relation to dip-ramp-plateau is the prediction noted above that the ramp will be unaltered in the range $0 < \mu < \mu_r$ for some $\mu_r$. Indeed we will find that this is a feature of our exact solution for the global scaling limit of $S_N(k; t)$ in the range $t>4$.
 
 \begin{remark} $ $ \\
1.~ Adding the square of (\ref{X.1}) with $k=2 \sqrt{2N} \tau_b $ to the asymptotic form of $S_N^{(\rm G)}(2 \sqrt{2N} \tau_b )$
 as implied by (\ref{Ma3.1}) gives a graphically accurate approximation to
 \begin{equation}\label{Ma3.2a}
 \Big \langle \Big | \sum_{l=1}^N e^{i 2 \sqrt{2N} \tau_b x_l} \Big |^2 \Big \rangle_{\rm GUE}
 \end{equation}
 as defined by the GUE analogue of (\ref{S.5}). A numerical plot --- see Figure \ref{F0} --- using a log-log scale of each axis exhibits the dip-ramp-plateau effect. \\
 2.~In general Dyson Brownian
 motion on $U(N)$ is dependent on the initial condition, with the particular choice of the identity matrix and thus all eigenvalues having angle $x_l(0)=0$ being the
 subject of the present work. From the above discussion, relevant questions in relation to dip-ramp-plateau of a more general choice (say with an eigenvalue density supported on an interval strictly
 within $(-\pi,\pi)$) are the functional form of the singularity of the boundary of support for $t > 0$, and the existence of a time such that the eigenvalue density is strictly positive for all angles in $(-\pi,\pi]$. \\
 3.~There is some interest in the functional form of the amplitude $A(t)$ in (\ref{Ma2}) and exponent $c(t)$ in (\ref{Ma2b}) from the viewpoint of the dip-wavenumber discussed
 in the paragraph below (\ref{Ma8}). First, for $t < 4$ we see by equating 
 the RHS of (\ref{Ma7}) multiplied by $N$ with (\ref{Ma2}) that a refinement of the asymptotic bound
  $k \asymp N^{1/2}$ for the dip-wavenumber is to include the dependence on $t$ by way of the amplitude $A(t)$ to obtain $k \asymp (A(t) N)^{1/2}$. The explicit functional
  form of $A(t)$ given in (\ref{E.1}) below shows that $A(t) \asymp (1-t/4)^{-1/4}$ for $t \to 4^-$. This divergence indicates a breakdown in the asymptotic dependence on $N$ for $t \ge 4$.
  In fact at $t=4$, the dip-wavenumber relates to $N$ through the asymptotic relation $k \asymp N^{6/11}$ --- see Remark \ref{R4.8} below. In the case $t > 4$, equating
   the RHS of (\ref{Ma7}) multiplied by $N$ with (\ref{Ma2b}) we see that by including the exponent $c(t)$ in the former the dip-wavenumber dependence on both
   $t$ and $N$ reads $k \asymp {1 \over c(t)} \log N$. The exponent $c(t)$ is given explicitly as $t \gamma(4/t)$ in (\ref{4.0d+}) below. We calculate from (\ref{gg+}) that
   for $t \to 4^+$, $c(t) \sim {8 \over 3}(1-4/t)^{3/2}$. Hence its reciprocal diverges in this limit, again in keeping with the distinct dip-wavenumber asymptotic relation for $t=4$.
   In the limit $t \to \infty$ we read off from (\ref{gg+}) that $c(t) \sim t$ and thus the dip-wavenumber asymptotic relation $k  \asymp {1 \over t} \log N$. The factor
   $1/t$ acts as a damping of the  dip-wavenumber for large $t$. This is in keeping with their being no dip effect in the $t \to \infty$ state of Dyson Brownian motion ---
   referred in random matrix theory as the circular unitary ensemble (CUE); see e.g.~\cite[Ch.~2]{Fo10} --- due to the eigenvalue density then being rotationally
   invariant.
   \end{remark}
 
\begin{figure*}
\centering
\includegraphics[width=0.75\textwidth]{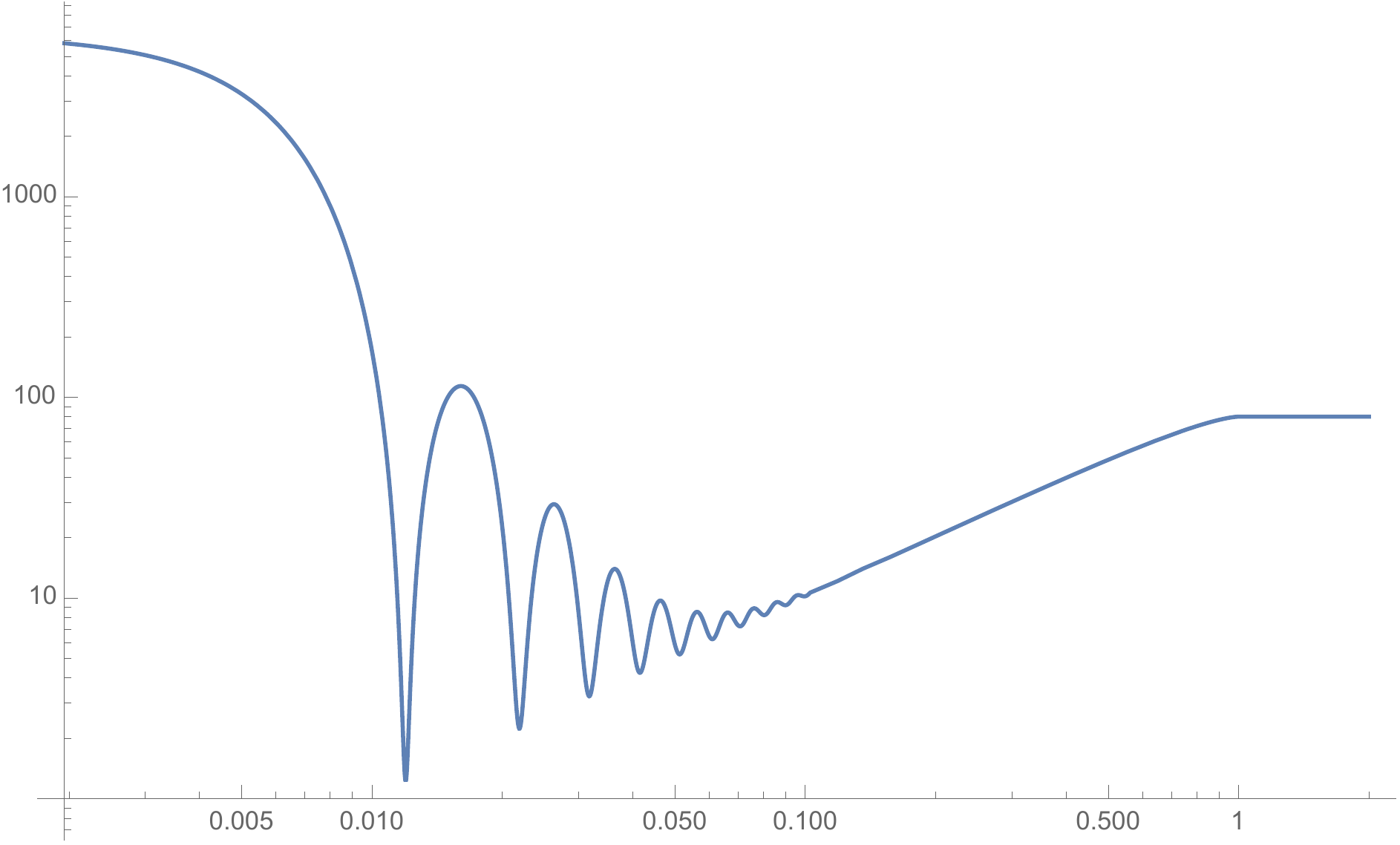}
\caption{Log-log plot of (\ref{Ma3.2a}) with $N = 80$ as a function of $\tau_b$.}
\label{F0}
\end{figure*}

\subsection{Layout of the paper and main results}
Section \ref{S2} presents a self contained derivation of the joint eigenvalue PDF for  Dyson Brownian
 motion on $U(N)$ from the identity. As noted in the recent work \cite{KLZF20}, the functional form can be identified as
 an example of a cyclic  Poly\'a ensemble. The significance of this is that the correlation kernel determining
 the general $k$-point correlation then admits a special structured form --- the corresponding theory is revised in
 Section \ref{S3}. 
 
 Section \ref{S4} relates to the density $\rho_{(1),N}(x;t)$ and its moments. Since the former is
a periodic function of $x$, period $2 \pi$, and even in $x$ it admits the Fourier expansion
\begin{align}\label{4.0a}
{2 \pi \over N} \rho_{(1), N}(x;t)   =  \sum_{k=- \infty}^\infty m_{k}^{(N)}(t) e^{ i k x}
 =  1 + 2 \sum_{k=1}^N  m_{k}^{(N)}(t) \cos  k x,
\end{align} 
where the $\{  m_{k}^{(N)}(t) \}$ are real and specified by
$$
   m_{k}^{(N)}(t) :=  {1 \over N} \int_{-\pi}^\pi \rho_{(1), N}(x;t) \, e^{- i k x} \, \mathrm dx.
   $$
In this context, the Fourier coefficients $\{  m_{k}^{(N)}(t) \}$ are referred to as the moments of the eigenvalue density. 
Our first new result is to detail the derivation of an exact evaluation of the moments, stated in some early literature,
but without derivation.

\begin{proposition} (Onofri \cite{On81} and Andrews and Onofri \cite{AO84}, both after correction)
We have
\begin{equation}\label{4.0c}
 m_{k}^{(N)}(t)   =  q^{k (N + k + 1)} \, {}_2 F_1 ( 1-N, 1 - k; 2; 1 - q^{-2k} ),
 \end{equation}
where ${}_2F_1(a,b;c;z)$ denotes the Gauss hypergeometric function and $q = e^{-t/2N}$. 
\end{proposition}

In fact two different derivations of (\ref{4.0c}) are given. The first involves Schur function averages,
and is the one stated in \cite{On81, AO84} as giving rise to (\ref{4.0c}) but with the details omitted. 
A novel aspect of our presentation is the use of the cyclic P\'olya ensemble structure to compute
the Schur function average.
 The second is to make use of the
cyclic P\'olya ensemble structure for the density itself. This is both quicker and more straightforward
than using Schur function averages.
The exact formula (\ref{4.0c}) allows for the large $N$ asymptotic form of the $m_{k}^{(N)}(t)$ to be 
determined in the regime that $k/N = \mu > 0$ is fixed.

\begin{coro}\label{C1a}
Fix $k/N = \mu > 0$. Define
\begin{equation}\label{c.1}
t^* = t^*(\mu) := {2 \over \mu} \log \Big |{1 + \mu \over 1 - \mu} \Big |.
\end{equation}
For $0 < t < t^*$ we have 
\begin{multline}\label{c.2}
 m_{k}^{(N)}(t)   = {(-1)^N \over N^{3/2}} \sqrt{2 \over \pi} e^{- \mu t /2} 
 {((1 - e^{-\mu t}) \mu)^{-1/2} \over ((1 - e^{-\mu t}) ( (\mu+1)^2 e^{- \mu t} - (\mu - 1)^2))^{1/4}}\\
 \times \cos(N \tilde{h}(t,\mu) +\pi/4) +
 {\rm O} \Big ( {1 \over N^{5/2}} \Big ),
 \end{multline}
 where $ \tilde{h}(t,\mu)  = h(\lambda,\mu) |_{\lambda = e^{- t \mu/2}}$ with $ h(\lambda,\mu)$ given
 by (\ref{h0}) below.
 
 On the other hand, for $ t > t^*$ we have that $ m_{k}^{(N)}(t) $ decays exponentially fast in $N$.
\end{coro}
The proof is given in Section \ref{S4.3}. The relevance of Corollary \ref{C1a} to the slope regime comes about by its
validity in an extended regime $k,N \to \infty$ with $k \ll N$, which is also proved in Section \ref{S4.3}. 

\begin{coro}\label{C1b}
Suppose $t < 4$, and let $A(t)$ be the amplitude of the square root singularity at the endpoints of the support
of $\rho_{(1),N}(x;t)$; recall the text below (\ref{Ma1}).
For $k, N \to \infty$ and $k \ll N$ we have
\begin{equation}\label{h0+2}  
  m_k^{(N)}(t) \sim {\sqrt{\pi} \over (N \mu)^{3/2}}  A(t) \cos ( k L_0(t) -  3 \pi/4 + {\rm O}(k^2/N)),
   \end{equation}  
 as is in agreement with (\ref{dy}).
  \end{coro}

The topic of Section \ref{S5} is the calculation of the spectral form factor $S_N(k;t)$, and its $N \to \infty$
asymptotic limits, first for $k$ fixed, and then for $k$ proportional to $N$. In Proposition \ref{P5.2}, Eq.~(\ref{S.15}),
$S_N(k;t)$ is expressed in terms of the an integral over a variable $s$, where the key factor in the integrand can be
identified with $(m_k^{(N)}(s+t))^2$, known explicitly according to (\ref{4.0c}). From this, the  $N \to \infty$ limit is
almost immediate. Further simplification then leads to our first limit formula in relation to the structure function.

  \begin{theorem}\label{P5.6m}
  We have
\begin{equation}\label{S.18x}      
 \lim_{N \to \infty} S_N(k;t) = k - e^{-kt} \sum_{s=0}^{k-1} (k-s)    \Big ( L_{s}^{(-1)}(kt) \Big )^2,
 \end{equation}
 where $L_s^{(a)}(z)$ denotes the Laguerre polynomial. 
 \end{theorem}
 
 Next considered is the regime that $k/N =: \mu$ is fixed for $N \to \infty$, as is relevant for dip-ramp-plateau.
 The same strategy as used to deduce Corollary \ref{C1a}, which proceeds by rewriting the ${}_2 F_1$ function
 in (\ref{4.0c}) in terms of a Jacobi polynomial with parameters $(N(\mu - 1), 1)$,  then makes use of known asymptotics 
 of the latter reported in the literature
 \cite{SZ22}, suffices to deduce the corresponding limit theorem.

     \begin{figure*}
\centering
\includegraphics[width=0.75\textwidth]{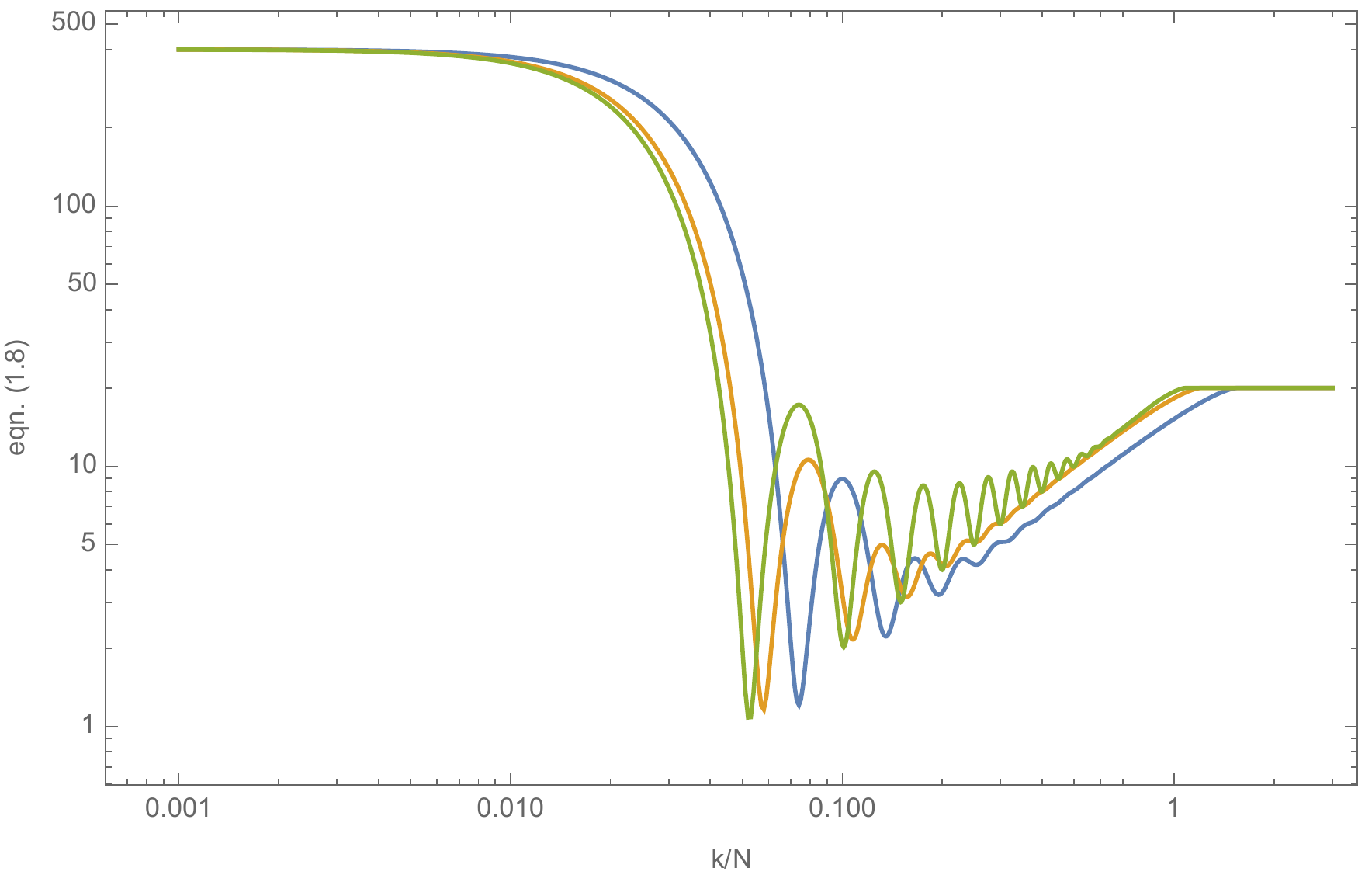}
\caption{[Colour online] Log-log plot of (\ref{S.5}), converted to a continuous function of $\mu = k/N$ for large $N$ according
to the procedure of the text, with $N = 20$ and various values of $t$. In order of the maximum height of the oscillations,
or equivalently sharpest dip,
these are $t=6$ (green), $t=4$ (yellow) and $t=2$ (blue). The (continuous) dip-ramp-plateau effect is evident; cf.~Figure \ref{F0}.}
\label{Dfig2}
\end{figure*}
 
    \begin{theorem}\label{P5.6}
    Define $t^*=t^*(\mu)$ by (\ref{c.1}) and require that $\mu > 0$. For $0 < t < t^*$ we have
 \begin{multline}\label{S.20a}      
\tilde{S}_\infty(\mu;t):=  \lim_{N \to \infty} {1 \over N} S_N(k;t) \Big |_{\mu = k/N}  \\
=  \min (\mu, 1)  - {\mu^3 \over \pi (\mu+1)} e^{- \mu t} \int_0^{(t^* - t)_+}
  { s e^{- \mu s}    \over (1 - e^{-\mu (s+t)})^{3/2}} {1 \over \sqrt{ e^{- \mu (s+t)} -  e^{- \mu t^*} }} \, \mathrm ds,
    \end{multline} 
    where $  (t^* - t)_+ =   t^* - t$ for $ t^* - t > 0$, and $ (t^* - t)_+ = 0$ otherwise.
    \end{theorem}

 As must be, $\tilde{S}_\infty(\mu;t)$ is well defined for continuous $\mu > 0$.
 Less immediate, but similarly true, is that this is a feature of the
 asymptotic formula (\ref{c.2}) for $m_k^{(N)}(t)$, even though for finite $N$ the Fourier
 coefficient label $k$ must be discrete. Consequently a continuous version of dip-ramp-plateau
 can be defined, which is the (log-log) plotting of (\ref{S.5}) as a function of $\mu$,
 with the first term on the RHS
 replaced by $N$ times the limit formula of Theorem \ref{P5.6}, and the second term
 replaced by $N^2$ times the large $N$ asymptotics of $( m_k^{(N)}(t) )^2$. 
 From a practical viewpoint of performing the plot, since the latter has 
 not been made explicit for $t \ge 4$, we replace this term by the finite $N$ Jacobi
 polynomial form (\ref{A.8b}) below with $k=\mu N$, as this prescription leads to the same
 asymptotic formulas; see Figure \ref{Dfig2} for some examples.

\section{The joint eigenvalue PDF}\label{S2}
Dyson showed that  Brownian motion associated with the group $U(N)$ gives that the eigenvalue PDF,
$p_\tau = p_\tau(e^{i x_1}, \dots, e^{i x_N})$ --- this quantity is a function too of the initial conditions --- evolves
according to the Fokker-Planck equation (referred to by Dyson as the Smoluchowski equation)
\begin{equation}\label{2.0a}
\gamma {\partial  p_\tau \over \partial \tau} = {\mathcal L} p_\tau, \qquad \mathcal L = \sum_{j=1}^N {\partial \over \partial x_j} \Big (
{\partial W \over \partial x_j}   + {1 \over \beta} {\partial \over \partial x_j} \Big ),
\end{equation}
where $\gamma$ is a scale for the time like parameter $\tau$, $\beta = 2$, and
\begin{equation}\label{2.0b}
W = - \sum_{1 \le j < k \le N} \log | e^{i x_k} -  e^{i x_j}  |.
\end{equation}
As pointed out in \cite{Dy62b}, this equation permits the interpretation of an $N$ particle system on the unit circle, with
particles interacting pairwise via the potential $- \log | e^{i x} - e^{i y} |$, and executing overdamped Brownian motion
in a fictitious viscous fluid with friction coefficient $\gamma$ at inverse temperature $\beta$.
he statistical properties of the eigenvalues at a particular time like parameter $\tau$ are completely determined by
$p_\tau$, which in turn requires solving (2.1) subject to a prescribed initial condition. This is an essential point of Dyson's original
work [32]. With the Fokker-Planck equation written in its equivalent form as coupled stochastic differential equations, this relation was
further studied in \cite{CL01}; see too \cite[\S 2.2]{BF22}.

A fundamental result of Sutherland \cite{Su71a} identifies a similarity transformation that maps the Fokker-Planck operator
$\mathcal L$ to a Schr\"odinger operator $H$. Thus we have
\begin{equation}\label{2.1a}
- e^{ \beta W/2}  \mathcal L  e^{- \beta W/2}  = {2 \over \beta} ( H - E_0)
\end{equation}
with
\begin{equation}\label{2.1b}
 E_0  =    \Big ( { \beta \over 2} \Big )^2   {N (N^2 - 1) \over 24}
\end{equation}
and
\begin{equation}\label{2.1c}
H = - {1 \over 2} \sum_{j=1}^N {\partial^2 \over \partial x_j^2} + (\beta/2)  (\beta/2 - 1) \sum_{1 \le j < k \le N}
{1 \over (2 \sin ( x_k - x_j)/2)^2}.
\end{equation}
 Notice that for $\beta = 2$, as required for Dyson Brownian motion on
$U(N)$, the interaction term in (\ref{2.1c}) vanishes.

In light of (\ref{2.1a}) and (\ref{2.1c}), knowledge of the free fermion Green function solution of the imaginary time Schr\"odinger equation on
a circle allows for the computation of $p_\tau$ in the case of the initial condition
\begin{equation}\label{3.0c1}
p_\tau(\mathbf x) \Big |_{\tau = 0}  = \prod_{l=1}^N \delta (x_l - x_l^{(0)}), \qquad (- \pi  <  x_1^{(0)} < \cdots
< x_N^{(0)} \le  \pi).
\end{equation}
Let us write $p_\tau(\mathbf x; \mathbf x^{(0)})$ to indicate this initial condition.
We will compute the functional form of $p_\tau(\mathbf x; \mathbf 0)$ by
first calculating $p_\tau(\mathbf x; \mathbf x^{(0)})$ in terms of a determinant
and then taking the limit $ \mathbf x^{(0)} \to \mathbf 0$.
In preparation, introduce the Jacobi theta functions
\begin{equation}\label{3.0c2}
\theta_2(z;q) := \sum_{n=-\infty}^\infty q^{(n - 1/2)^2} e^{2 i z (n - 1/2)}, \qquad
\theta_3(z;q) := \sum_{n=-\infty}^\infty q^{ n^2} e^{2 i z  n}.
\end{equation}
Scale $\tau$ by setting $ \tau / \gamma = t/N$, and relate
$q$ in (\ref{3.0c2}) to $t/N$ be setting
\begin{equation}\label{3.0c3}
q = e^{-t/2N}.
\end{equation}

\begin{proposition} (Liechty and Wang \cite{LW16},  Kieburg et al. \cite{KLZF20})
With  $\kappa = 2$ for $N$ even and $\kappa = 3$ for $N$ odd and $q$ as in (\ref{3.0c3}) we have 
\begin{equation}\label{3.0c5}
p_t(\mathbf x; \mathbf 0) = { q^{N (N^2 - 1)/12} \over (2 \pi)^N \prod_{l=1}^{N } l!} \Big ( \prod_{1 \le j < k \le N} \sin  (x_k - x_j)/2\Big )
\det \Big [ \Big (- {2} {\partial \over \partial x_j }\Big  )^{k-1}  \theta_\kappa ( x_j/2; q) \Big ]_{j,k=1,\dots,N}.
\end{equation}
\end{proposition} 

\begin{proof}
The solution of the  single particle imaginary time Schr\"odinger equation
$$
\gamma {\partial \over \partial \tau} g_\tau = - {1 \over 2} {\partial^2 \over \partial x^2} g_\tau, \qquad - \pi < x <  \pi
$$
with initial condition $ g_\tau |_{x = x^{(0)}} \to \delta(x - x^{(0)})$ as $\tau \to 0^+$  is
$$
 g_\tau(x; x^{(0)}) =  \left \{ \begin{array}{ll}  {1 \over 2 \pi} \theta_3 ( (x - x^{(0)})/2; q), & {\rm periodic \: b.c.} \\
 {1 \over 2 \pi}   \theta_2( (x - x^{(0)})/2; q), &  \text{anti-periodic  b.c.} \end{array} \right., 
 $$
 where $q=e^{-\tau/2 \gamma}$. This can be verified directly, with the initial condition following from the functional
 form of the theta functions written in conjugate modulus form; see \cite{WW65}.
 
 Forming a Slater determinant gives that the Green function free fermion solution of
 the $N$-particle imaginary time free Schr\"odinger equation 
 $$
 \gamma {\partial \over \partial \tau} g_\tau = - {1 \over 2} \sum_{j=1}^N {\partial^2 \over \partial x_j^2} g_\tau, \qquad - \pi < x_j < \pi
$$
is
\begin{equation}\label{3.0c4m}
 g_\tau(\mathbf x; \mathbf x^{(0)}) =  \left \{ \begin{array}{ll} \displaystyle  \det \Big [ {1 \over 2 \pi } \theta_3((x_j - x_k^{(0)})/2; q) \Big ]_{j,k=1}^N, &   {\rm periodic \: b.c.}  \\[3mm]
 \displaystyle \det \Big [ {1 \over 2 \pi} \theta_2((x_j - x_k^{(0)})/2; q) \Big ]_{j,k=1}^N, &  \text{anti-periodic  b.c.} \end{array} \right..
 \end{equation}
 The significance of knowledge of $g_\tau(\mathbf x; \mathbf x^{(0)})$ is that it follows from (\ref{2.1a}) with $\beta = 2$ that
 $$
 p_\tau( \mathbf x;\mathbf x^{(0)})   =  e^{E_0  |_{\beta = 2}\tau/ 2 \gamma}  \bigg ( {\prod_{1 \le j < k \le N} \sin  ((x_k - x_j)/2) \over \prod_{1 \le j < k \le N} \sin  ((x_k^{(0)} - x_j^{(0)})/2) }  \bigg )  g_\tau( \mathbf x; \mathbf x^{(0)}). 
 $$
Consequently, requiring that $ p_\tau( \mathbf x;\mathbf x^{(0)}) $ exhibits periodic boundary conditions with respect to the translation $x_j \mapsto x_j + 2 \pi $ and setting $\tau/\gamma = t/N$
we  have \cite{PS91,Fo90b,Fo96}
\begin{equation}\label{3.0c4}
p_t(\mathbf x; \mathbf x^{(0)})    =    q^{N (N^2 - 1)/12}   \bigg ( {\prod_{1 \le j < k \le N} \sin  ((x_k - x_j)/2) \over \prod_{1 \le j < k \le N} \sin  ((x_k^{(0)} - x_j^{(0)})/2) }  \bigg )
\det \Big [ {1 \over 2 \pi } \theta_\kappa( ( x_j - x_k^{(0)})/2); q ) \Big ]_{j,k=1,\dots,N},
\end{equation}
where $\kappa = 2$ for $N$ even and $\kappa = 3$ for $N$ odd and now $q$ is given by (\ref{3.0c3}).
 Here the different functional forms depending on the parity of $N$
can be traced back to the product over pairs in the numerator of the RHS of (\ref{3.0c4}) being multiplied by the sign $(-1)^{N - 1}$
upon the mapping $x_j \mapsto x_j + 2 \pi $. With $\theta_2$ for $N$ even, and $\theta_3$ for $N$ odd, the determinant factor on
the RHS of (\ref{3.0c4}) has the same property as noted in (\ref{3.0c4m}), implying that $p_\tau$ itself is always periodic. Also, in keeping with the 
ordering in (\ref{3.0c1}) and thus the implied normalisation, the normalisation associated with (\ref{3.0c4}) is so that
$ \int_R p_\tau(\mathbf x) \, \mathrm d \mathbf x = 1$, where the region $R$ is specified by 
$- \pi  < x_1 < \cdots < x_N <  \pi$.

 Application of L'H\^{o}pital's rule in (\ref{3.0c4}) to take $x_j^{(0)} \to 0$ for $j=1,2,\dots,N$ in succession
gives (\ref{3.0c5}).
A further factor of $1/N!$ relative to (\ref{3.0c4}) has been included to allow
the normalisation condition to become $ \int_{{[-\pi,\pi]}^N} p_t(\mathbf x; \mathbf 0)  \, \mathrm d \mathbf x = 1$.
\end{proof}

\section{Cyclic P\'olya ensemble structure}\label{S3}
\subsection{Definition of a cyclic P\'olya ensemble}
The eigenvalue PDF (\ref{3.0c5}) has the  structural property of being proportional to
\begin{equation}\label{5.1a}
\Big ( \prod_{1 \le j < k \le N} \sin (x_k - x_j)/2 \Big )
\det \Big [ {\partial^{k-1} \over \partial x_j^{k-1}} \hat{w}(x_j) \Big ]_{j,k=1,\dots,N},
\end{equation} 
for a particular  $ \hat{w}(x)$. Such a form has been isolated in the recent work \cite{KLZF20}. This was in the context of a study of the implications of
the theory of matrix spherical transforms on $U(N)$, as applied to multiplicative convolutions conserving a
determinant structure, first presented in \cite{ZKF21}.

By setting $z_j = e^{ i x_j}$ it is easy to see that (\ref{5.1a}) has the equivalent complex form 
\begin{equation}\label{5.1b}
\frac{1}{Z_N}\Delta_N(\mathbf z) \det \Big [ \Big ( - z_j {\partial  \over \partial z_j} \Big )^{k-1} w(z_j) \Big ]_{j,k=1,\dots,N}.
\end{equation} 
Here $w(z) = z^{-(N-1)/2} \hat{w}(z)$, $Z_N$ is the normalisation constant which can be determined for general $N$
(see~\eqref{normalisation} below), and $\Delta_N(\mathbf z)$, which is referred to as the Vandermonde product, is specified by
\begin{equation}\label{5.1c}
\Delta_N(\mathbf z) := \prod_{1 \le j < k \le N}(z_k - z_j).
\end{equation}
The range of $\mathbf z$ is 
\begin{equation}
	\mathbf z\in\mathbb S_1^N:=\{z\in\mathbb C^N: \forall j= 1,\ldots,N, |z_j|=1\},
\end{equation}
where the permutation invariance relaxes the ordering of eigenvalues. By requiring $w(z)$ to be a $(2M+\chi-1)$-differentiable cyclic P\'olya frequency function of order $N$ (odd with $\chi=1$ or even with $\chi=0$), multiplied by $z^{-M-\chi+1}$~\cite[Sec. 2.5]{ZKF21}, one can ensure the positivity of~\eqref{5.1b}, and hence specify an ensemble with its eigenvalue PDF of this general form; \eqref{5.1b} is called a \textit{cyclic P\'olya ensemble}. For technical purposes, we further restrict the weight function $w$ to be in the set
\begin{equation}
	\widetilde{L}^1(\mathbb S_1):=\{w\in L^1(\mathbb S_1):w(z)^*=z^{N-1}w(z)\text{ and }\partial_z^jw(z)\in L^1(\mathbb S_1)\text{ for }j=0,\ldots,N-1\};
\end{equation}
see \cite{KLZF20} for further details.
This allows us to construct bi-orthogonal systems of functions corresponding to each cyclic P\'olya ensemble. However, there are further aspects of the
theory that need to be revised before doing this.

\subsection{Spherical transform and closure under multiplicative convolution}\label{S2.2}
One distinctive feature of cyclic P\'olya ensembles is that they are closed under multiplicative convolution.
Thus a product of two random matrices drawn from this class of ensembles is still a random matrix drawn from this class of ensembles. To revisit the underlying theory
\cite{KLZF20}, we first  introduce the spherical transform of a $ U(N)$ random matrix with eigenvalue PDF $f$. This is defined as
\begin{equation}\label{2.6}
	\mathcal S f(s):=\int_{\mathbb S_1^N}\left(\prod_{j=1}^N\frac{\mathrm d z_j}{2\pi iz_j}\right)f(z)\Phi(z,s),\quad s_j\ne s_k\text{ for any }j\ne k, \: s_j \in \mathbb Z,
\end{equation}
where $\frac{\mathrm d z_j}{2\pi iz_j}$ is the uniform measure on the $j$-th unit circle $\mathbb S_1$, and $\Phi$ is the normalised character of irreducible representations of $U(N)$. 
The latter has the explicit form
\begin{equation}
	\Phi(z,s):=\left(\prod_{j=0}^{N-1}j!\right)\frac{\det[z_j^{s_k}]_{j,k=1}^N}{\Delta_N(\mathbf z)\Delta_N( \mathbf s)}.
\end{equation}

There are two important properties of the spherical transform. First, it possesses an inversion formula~\cite[(2.2.12)]{KLZF20}, and therefore every $\mathcal S f$ is uniquely associated with the eigenvalue PDF $f$. Second, it also admits a convolution formula. Thus, for independent random matrices $U_1,U_2$ with eigenvalue PDFs $f_{U_1}$ and $f_{U_2}$, the spherical transform of the eigenvalue PDF $f_{U_1U_2}$ of the product $U_1U_2$ has the factorisation
\begin{equation}\label{convolution_formula}
	\mathcal S f_{U_1U_2}(s)=\mathcal S f_{U_1}(s)\mathcal S f_{U_2}(s).
\end{equation}
These two properties allow us to obtain the eigenvalue PDF of the product $U_1U_2$ by inverting the right side of~\eqref{convolution_formula}.

We are now in a position to address the question of closure with respect to multiplicative convolution. First of all, integrating~\eqref{5.1b} shows that the normalisation constant is given by
\begin{equation}\label{normalisation}
	Z_N= N! \prod_{j=0}^{N-1}j!\mathcal S w(j).
\end{equation} 
Here $\mathcal S w$ is the one-dimensional spherical transform of the weight function $w$, which correspond to the coefficients of the Fourier series expansion of $w$. To show closure  under multiplicative convolution, an explicit calculation tells us that the spherical transform of a cyclic P\'olya ensemble of the form~\eqref{5.1b} is given by
\begin{equation}\label{sf_cpe}
	\mathcal Sf(s)=\prod_{j=1}^{N}\frac{\mathcal Sw(s_j)}{\mathcal Sw(j-1)},
\end{equation}
(see~\cite[Prop. 4]{KLZF20}). From the convolution formula, it is now not difficult to see that
\begin{equation}\label{sf_cpe_product}
	\mathcal Sf_{U_1U_2}(s)=\prod_{j=1}^{N}\frac{\mathcal Sw_{U_1}(s_j)\mathcal Sw_{U_2}(s_j)}{\mathcal Sw_{U_1}(j-1)Sw_{U_2}(j-1)}=\prod_{j=1}^{N}\frac{\mathcal S[w_{U_1}\ast w_{U_2}](s_j)}{\mathcal S[w_{U_1}\ast w_{U_2}](j-1)},
\end{equation}
where $w_{U_1}\ast w_{U_2}$ denotes the univariate multiplicative convolution of $w_{U_1}$ and $w_{U_2}$. Explicitly,
\begin{equation}\label{2.11}
	w_{U_1}\ast w_{U_2}(z):=\int_{\mathbb S_1}w_{U_1}(zy^{-1})w_{U_2}(y) {\mathrm d y \over 2 \pi i}.
\end{equation}
As~\eqref{sf_cpe_product} has the same structure as~\eqref{sf_cpe}, one concludes that the product $U_1U_2$ is also a cyclic P\'olya ensemble with a new weight function $w_{U_1}\ast w_{U_2}$.

\subsection{Biorthogonal system}
Another important feature of the cyclic P\'olya ensemble is that it corresponds to a determinantal point process, which
in turn can be characterised by a bi-orthogonal system.
With the formula for the ratio of gamma functions
\begin{equation}
	\frac{\Gamma(j-l)}{\Gamma(-l)}=(-1)^{j-1}\frac{\Gamma(l+1)}{\Gamma(l-j+1)}, \qquad l \in \mathbb Z,
\end{equation}
we introduce a set of pairs of functions $\{(P_j,Q_j)\}_{j=0,\ldots,N-1}$, where
\begin{equation}\label{biorth-cyc.Pol}
\begin{split}
P_j(z_1):=&\sum_{k=0}^j\frac{1}{(j-k)!k!}\frac{(-z_1)^k}{\mathcal{S}w(k)},\\
Q_j(z_2):=&z_2\partial_{z_2}^j{z_2}^{j-1}w(z_2)=\lim_{\epsilon\to0^+}\sum_{l\in\mathbb Z\backslash\{0,1,\ldots,j-1\}}\frac{\Gamma(j-l)}{\Gamma(-l)}\mathcal{S}w(l)z_2^{-l}e^{-\epsilon (l+1-N)l}
\end{split}
\end{equation}
for $j=1,\ldots,N-1$. In the formula for $Q_j$ the factor $e^{-\epsilon (l+1-N)l}$ is a regularisation. This is required to ensure convergence of the infinite sum in the case of general weights $w(z)$.
 It is shown in~\cite{KLZF20} that both $\{P_j(z_1)\}_{j=0,\ldots,N-1}$ and  $\{z_1^{j-1}\}_{j=1,\ldots,N}$ span the same vector space, and so do both $\{Q_j(z_2)\}_{j=0,\ldots,N-1}$ and  $\{(z_2\partial_{z_2})^{j-1}w(z_2)\}_{j=1,\ldots,N}$. Moroever 
$\{P_j\}$ and $\{Q_j\}$ have been constructed to
form a bi-orthogonal set  with respect to the Haar measure on $\mathbb S_1$. Thus
\begin{equation}\label{2.15}
	\int_{\mathbb S_1}\frac{\mathrm d z}{2\pi i z}P_a(z)Q_b(z)=\delta_{a,b},
\end{equation}
where $\delta_{a,b}$ denotes the Kronecker delta function. For this reason, we say that $\{(P_j,Q_j)\}_{j=0,\ldots,N-1}$ is the corresponding bi-orthogonal system for the cyclic P\'olya ensemble. 

From the linearity of determinants in~\eqref{5.1b}, the span property of the polynomials $\{P_j(z_1)\}_{j=0,\ldots,N-1}$ and the functions  $\{Q_j(z_2)\}_{j=0,\ldots,N-1}$ 
enables us to rewrite the PDF as
\begin{equation}\label{2.14}
	{1 \over N!} \det[P_{j-1}(z_k)]_{j,k=1}^N\det[Q_{j-1}(z_k)]_{j,k=1}^N.
\end{equation}
We can check that (\ref{2.14}) 
is correctly normalised. For this we recall Andr\'eief's identity (see e.g.~\cite{Fo18}), which gives
\begin{multline*}
{1 \over N!} \int_{\mathbb S_1^N}\left(\prod_{j=1}^N\frac{\mathrm d z_j}{2\pi iz_j}\right)  \det[P_{j-1}(z_k)]_{j,k=1}^N\det[Q_{j-1}(z_k)]_{j,k=1}^N
\\ = \det \bigg [  \int_{\mathbb S_1}  \frac{\mathrm d z_j}{2\pi iz}  P_{j-1}(z) Q_{k-1}(z) \,\bigg ]_{j,k=1,\dots,N}.
\end{multline*}
Application of (\ref{2.15}) shows that the matrix on the RHS is in fact the identity, and thus the determinant is equal to unity, as required.

\subsection{Correlation kernel}
Associated with any bi-orthogonal system for a determinantal point process
 is its correlation kernel $K_N(z,z')$. By definition, the corresponding 
 $m$-point correlation function can be expressed as an $m$-dimensional determinant of this kernel,
\begin{equation}\label{2.16}
	\rho_{(m),N}(z_1,\ldots,z_m)=\det[K_N(z_j,z_k)]_{j,k=1}^m.
\end{equation}
Starting from the expression (\ref{2.14}), it is a standard fact that \cite{Bo98}
\begin{equation}\label{2.16a}
K_N(z_1,z_2) = \sum_{j=1}^N  P_{j-1}(z_1) Q_{j-1}(z_2).
\end{equation}
It is shown  in~\cite{KLZF20} that with the substitution (\ref{biorth-cyc.Pol}), (\ref{2.16a}) permits the rewrite
\begin{multline}\label{kernel-cyc.Pol}
K_N(z_1,z_2)=\sum_{k=0}^{N-1}(z_1z_2^{-1})^k \\
+\lim_{\epsilon\to0^+}\sum_{k=0}^{N-1}\sum_{l\in\mathbb Z\backslash\{0,\ldots,N-1\}}\frac{\Gamma(N-l)}{\Gamma(-l)\Gamma(N-k)\Gamma(k+1)}\frac{\mathcal{S}w(l)}{\mathcal{S}w(k)}\frac{(-z_1)^{k}z_2^{-l}}{k-l}e^{-\epsilon (l+1-N)l},
\end{multline}
where as in the second expression of (\ref{biorth-cyc.Pol}) the factor involving $\epsilon$ is included to ensure convergence of the infinite sum for general weights $w(z)$.

\begin{remark}
Using the identity
$$
\prod_{1 \le j < k \le N} \sin (x_k - x_j)/2 \propto \prod_{l=1}^N z_l^{-(N-1)/2} \Delta_N(\mathbf z),
$$
as implicit in the passage from (\ref{5.1a}) to (\ref{5.1b}), shows that for general initial conditions $\mathbf z^{(0)}$ the PDF (\ref{3.0c4}) has the form
(\ref{2.14}). From the general theory of biorthogonal ensembles in random matrix theory \cite{Bo98} the correlation functions are again of the form (\ref{2.16}), albeit
with $K_N$ now dependent on $\mathbf z^{(0)}$; see \cite{VP09} for results in this direction.
\end{remark}

\subsection{Application to the PDF \eqref{3.0c4}}
We conclude this section by presenting the cyclic P\'olya ensemble structure and the correlation kernel for~\eqref{3.0c4}, which is a direct corollary of the theory presented in~\cite{KLZF20},
as revised above.

\begin{coro}[{\cite[special case of Prop.~19]{KLZF20}}]
	With the change of variable $z_j = e^{i  x_j}$, the PDF~\eqref{3.0c5} becomes a cyclic P\'olya ensemble with weight function
	\begin{equation}\label{2.18}
		w(z)=\sum_{s=-\infty}^\infty q^{(s-(N-1)/2)^2}z^{-s}, \qquad q=e^{-t/2N}.
	\end{equation}
		The correlation kernel corresponding to~\eqref{3.0c5} is
	\begin{equation}\label{kernel}
	\begin{split}
	K_N(x,y)=& \frac{1}{2 \pi } \bigg(\frac{e^{ iN(x-y)}-1}{e^{ i(x-y)}-1}\\&+\sum_{k=0}^{N-1}\sum_{l\in\mathbb Z\backslash\{0,\ldots,N-1\}}
	 \frac{\Gamma(N-l)q^{(l-(N-1)/2)^2-(k-(N-1)/2)^2}}{\Gamma(-l)\Gamma(N-k)\Gamma(k+1)}\frac{(-1)^{k}e^{ i(xk-yl)}}{k-l}\bigg).
	\end{split}
	\end{equation}
	Compared to~\eqref{kernel-cyc.Pol}, here the regularisation is removed.
\end{coro}

\begin{proof}
	The spherical transform of $w$ is read off from the coefficients of its Fourier series form (\ref{2.18}). Thus
\begin{equation}\label{sf_cpe+}
\mathcal S w(s)=q^{(s-(N-1)/2)^2}. 
\end{equation}
With $z_j = e^{ix_j}$, upon recalling 
	(\ref{normalisation}), we see that~\eqref{5.1b} can be rewritten as
	\begin{equation}
		\frac{1}{q^{N(N^2-1)/12}\prod_{j=1}^{N}j!}\prod_{j<k}\left(e^{ix_k}-e^{ix_j}\right)\det\left[\left(i\frac{\partial}{\partial x_j}\right)^{k-1}w(e^{ix_j})\right]_{j,k=1,\ldots,N}.
	\end{equation}
	The Haar measure $\frac{\mathrm d z_j}{2\pi i z_j}$ is replaced accordingly by $\frac{\mathrm dx_j}{2\pi}$, which is responsible for the factor of
	$\frac{1}{2\pi}$ in (\ref{kernel}) since there we take $\mathrm dx_j$ as the reference measure. Now one notices that the Vandermonde product can be replaced by
	\begin{equation}
		\prod_{j<k}\left(e^{ix_k}-e^{ix_j}\right)= e^{(N-1)i \sum_{j=1}^N x_j/2 }\prod_{j<k}2i\sin\left(\frac{x_k-x_j}{2}\right).
	\end{equation}
	Also, it can be checked that the Jacobi theta functions~\eqref{3.0c2} have the unified form
	\begin{equation}
		\theta_{\kappa}\left(\frac{y}{2};q\right)=\exp\left(iy\frac{N-1}{2}\right)w(e^{iy}),
	\end{equation}
	for $N$ being even with $\kappa=2$ and $N$ being odd with $\kappa=3$. Therefore taking the derivatives of the theta functions and applying a row reduction to reduce every entry to be the highest derivative of $w$, one has
	\begin{equation}
		\det\left[\left(i\frac{\partial}{\partial x_j}\right)^{k-1}\theta_\kappa\left(\frac{x_j}{2};q\right)\right]_{j,k=1,\ldots,N}=
		e^{(N-1)i \sum_{j=1}^N x_j/2 }
		\det\left[\left(i\frac{\partial}{\partial x_j}\right)^{k-1}w(e^{ix_j})\right]_{j,k=1,\ldots,N}.
	\end{equation}
	This is~\eqref{5.1b}. To obtain the correlation kernel, one substitutes $z_j = e^{i x_j}$ into~\eqref{kernel-cyc.Pol}.
	
	%

	 It remains to show that the $l$-sum is absolutely convergent such that we can drop the limit $\epsilon\to0^+$ in (\ref{kernel}).
	  After decomposing the $l$-sum into two sums, where one  consists of positive $l$ and the other one is for negative $l$, we need to verify that both series with all positive terms
	\begin{equation}
		S_1:=\sum_{l=N}^\infty\frac{\Gamma(l+1)}{\Gamma(l-N+1)}\frac{q^{(l-\frac{N-1}{2})^2}}{l-k},\quad 
		S_2:=\sum_{\tilde l=1}^\infty\frac{\Gamma(N+\tilde l)}{\Gamma(\tilde l)}\frac{q^{(\tilde l+\frac{N-1}{2})^2}}{k+\tilde l}
	\end{equation}
	are convergent for $k\in\{0,\ldots,N-1\}$ and $q<1$. A ratio test suffices to give this claim.
\end{proof}

\begin{remark}\label{R3.2} $ $\\
1.~In relation to multiplicative convolution of two matrices $U_1, U_2$ drawn from the PDF~\eqref{3.0c5}, one with $q_1 = e^{-t_1/2N}$, and the other with $q_2 = e^{-t_2/2N}$,
      the theory in the final paragraph of Section \ref{S2.2} is applicable. Thus the PDF corresponding to $U_1 U_2$ is again a cyclic P\'olya ensemble with weight function
      computed from (\ref{2.11}), in keeping with the laws of unitary Brownian motion forming a semigroup under the multiplicative convolution
      $$
      f_N^{(U_1)} *  f_N^{(U_1)} = \int_{U(N)} f_N^{(U_1) }(U')  f_N^{(U_2) }(U {U'}^{-1})  \, d \mu(U'),
      $$
      where $f_N^{(U)}$ denotes the PDF.
       This gives back the same weight (\ref{2.18}), but with parameter $q = q_1 q_2$, or equivalently with time parameter $t = t_1 + t_2$. We remark that this in turn coincides with the fact that in the numerical construction (\ref{1.0}) the Brownian motions have Gaussian increments with variance depending only on the length of the time interval.\\
      2.~Substituting (\ref{sf_cpe+}) in the first equation of (\ref{biorth-cyc.Pol}) shows
 $$
 P_j(z) = q^{-(N-1)^2/4} \sum_{k=0}^j {1 \over (j - k)! k!} q^{k (N - 1 - k)} (-z)^k.
 $$
 One sees in particular that
 \begin{equation}\label{3.28}
 z^{-n/2} P_n(z) \Big |_{N = n + 1} \propto \sum_{k=0}^n (-1)^{n-k} \binom{n}{k} q^{k (n - k)} z^{(k - n/2)} =: \tilde{H}_n(z;q).
 \end{equation}
 The RHS defines the so-called unitary Hermite polynomials (terminology from \cite{Ka22}). They first appeared,
 in a random matrix context at least, in the thesis of Mirabelli \cite{Mi21}  on finite free probability \cite{Ma21}. They are
 extensively studied in the recent work of Kabluchko \cite{Ka22} from various viewpoints;
 see Remark \ref{R4.1a} below for more on this. We can use  (\ref{3.28}) to deduce a formula for $\tilde{H}_n(z;q)$ expressed
 as an average with respect to the PDF (\ref{3.0c5}).
 First, in relation to the polynomials $\{ P_j(z) \}$ generally, we observe the representation as an average $P_l(z) = \langle \prod_{s=1}^l (z - z_s)
 \rangle$. Here $\langle \cdot \rangle$ is with respect to the PDF (\ref{5.1b}), but where the size of the Vandermonde product and the determinant
 are reduced from $N$ to $l$, although with $w(z)$ unchanged and thus still dependent on $N$. Replace this latter $N$ by $n+1$,
 then set $l=n$.
 As in verifying the equality between (\ref{5.1a}) and (\ref{5.1b}), one can then check that the LHS of (\ref{3.28}),
 and thus up to proportionality $ \tilde{H}_n(e^{ix};q)$, admits the form
 $\langle \prod_{k=1}^n \sin ( ( x - x_k)/2) \rangle$, where the average is with respect to the PDF (\ref{3.0c5}) with $N = n$.

 \end{remark}


\section{The density and its moments}\label{S4}
\subsection{Known results}\label{S3.1}
The eigenvalue density, $\rho_{(1), N}(x;t)$ say, corresponding to (\ref{3.0c5}) is specified in terms of $p_t(\mathbf x; \mathbf 0)$ by
\begin{equation}\label{4.0}
\rho_{(1), N}(x;t) = N \int_{[-\pi,\pi]^{N-1}}  p_t(\mathbf x; \mathbf 0) \, \mathrm dx_2 \cdots dx_N.
\end{equation} 
Within random matrix theory the corresponding moments (\ref{4.0a}) 
were first considered
by  Biane \cite{Bi97} using  methods of free probability theory, and independently by Rains \cite{Ra97}.  Specifically the exact evaluation for $N \to \infty$
\begin{equation}\label{4.0b}
m_{k}^{(\infty)}(t)  := \lim_{N \to \infty} m_{k}^{(N)}(t)   = {e^{-kt/2} \over k} L_{k-1}^{(1)}(kt) =
e^{-kt/2} \, {}_1 F_1 ( 1 - k; 2; kt) , \qquad k \ge 1,
\end{equation} 
was obtained. Here
 $L_n^{(\mu)}(x)$ denotes the Laguerre polynomial of degree $n$ with parameter $\mu$,
and ${}_1 F_1(a;b;x)$ denotes the confluent hypergeometric function with parameters, $a,b$ and
argument $x$, which like the Laguerre polynomial can be defined by its power series expansion.

It turns out that the same result (\ref{4.0b}) can already be found in the literature on exactly solvable
low-dimensional field theories from the early 1980's \cite{Ro81} (see also \cite{Ka81} for the
cases $n=2,3$). The point here is that these field theories can be reduced to Dyson Brownian motion
on $U(N)$ (also referred to as the heat kernel measure on $U(N)$), with $\langle {\rm Tr} \, U^k \rangle$ an observable relating to Wilson loops.
Moreover, soon after the exact formula (\ref{4.0c}) for the finite $N$ case was found \cite{On81,AO84}, although the result was stated without
derivation and contains misprints; we reference
\cite[Eq.~(3.14) after correction]{GT14} for a more recent work relating to this result, also in a field theory context.
 From the series form of the Gauss hypergeometric function ${}_2 F_1(a,b;c;z)$ in (\ref{4.0c}),
 which with $b = 1 - k$ terminates to give a polynomial in $z$ of degree $k-1$, the second
 expression in (\ref{4.0b}) is reclaimed in the limit $N \to \infty$ upon recalling from
 (\ref{3.0c3}) that $q = e^{-t/2N}$.
 
 Defining
 \begin{equation}\label{4.0d}
 \rho_{(1), \infty}(x;t) := \lim_{N \to \infty} {2 \pi \over N}  \rho_{(1), N}(x;t) =  1 + 2 \sum_{k=1}^\infty m_k^{(\infty)}(t)  \cos  k x ,
  \end{equation}
  substituting (\ref{4.0b}) does not lead to a closed expression by way of evaluation of the summation over $k$.
  Instead, one source of analytic information comes from consideration of the large $k$ form of $m_k^{(\infty)}$.
  Thus it follows from (\ref{4.0b}) that for $t < 4$ the rate of decay is $O(k^{-3/2})$ and is consistent with (\ref{Ma1}) \cite{GM95}.
  This is no longer true for  $t \ge 4$. Explicitly, for $t > 4$ the asymptotic analysis of the Laguerre polynomial given in 
  \cite{GM95} implies
   \begin{equation}\label{4.0d+}
   m^{(\infty)}_k(t) \mathop{\sim}\limits_{k \to \infty} {t \over \sqrt{2 \pi k}} \Big ( 1 - {4 \over t} \Big )^{-1/4} e^{-k t \gamma(4/t)/2},
   \end{equation}
   where
  \begin{equation}\label{gg+}  
   \gamma(x) := \sqrt{1 - x} - {x \over 2} \log {1 + \sqrt{1 - x} \over 1 - \sqrt{1 - x}},
    \end{equation} 
 giving that the rate of decay is exponentially fast \cite{GM95}. At $t=4$ the decay is $O(k^{-4/3})$ \cite{Le08}.
  Moreover, the support of the density for $t < 4$
  is $[-L_0(t), L_0(t)]$ where \cite{DO81,Bi97}
  \begin{equation}\label{4.0e} 
  L_0(t) = {1 \over 2} \sqrt{t (4 - t)} + {\rm Arcos} \, (1 - t/2),
  \end{equation}
while for $t > 4$ the support is the full interval $[-\pi,\pi]$ with $  \rho_{(1), \infty}(x;t) \to 1$ as $t \to \infty$.  

Next we summarise results from  \cite{Bi97}, following the clear presentation in \cite{Ha18}. 
Set
$$
 z = e^{i x}, \quad \rho_{(1),\infty}(x;t) \mapsto \rho_{(1),\infty}(z;t),
$$
and introduce the Herglotz transform
 \begin{equation}\label{a.0}
 H_t(w) =  \int_{\mathcal C} {z + w \over z - w}   \rho_{(1),\infty}(z;t) \,{ \mathrm  dz \over 2 \pi i z}
 \end{equation}
 where the contour $\mathcal C$ is the unit circle in the  complex $z$-plane, traversed anti-clockwise.
 With
 \begin{equation}\label{a.1} 
 \psi_t(w) = 
\int_{\mathcal C} { w \over z - w}   \rho_{(1),\infty}(z;t) \,  {\mathrm dz  \over 2 \pi i z} =   \sum_{p=1}^\infty w^p m_p^{(\infty)}(t) 
 \end{equation}
 denoting the moment generating function one sees
 \begin{equation}\label{a.2}  
 H_t(w) = 1 + 2 \psi_t(w)
 \end{equation}
 and thus
  \begin{equation}\label{a.3}  
 \rho_{(1),\infty}(w;t)  = {\rm Re} \, H_t(w).
  \end{equation}
  From knowledge of $\{ m_p^{(\infty)}(t) \}$ from (\ref{4.0b}) substituted in (\ref{a.1}), with the result then substituted in
  (\ref{a.2}), we can check that $H_t$ satisfies the partial differential equation
   \begin{equation}\label{a.4}  
   \Big ( {\partial \over \partial t} + w H_{2t}(w) {\partial \over \partial w} \Big ) H_{2t}(w) = 0,
    \end{equation}
    subject to the initial condition $H_0(w) = (1 + w)/(1-w)$. The equation (\ref{a.4}) can
    be identified with  the complex Burgers equation in the inviscid limit, and
    can also be derived by other considerations \cite{PS91,FG16}. It can be checked from
    (\ref{a.4}) that $H_t$ satisfies the functional equation
     \begin{equation}\label{a.5}  
     {H_t(w) - 1 \over H_t(w) + 1} e^{(t/2) H_t(w)} = w.
      \end{equation}  

 Using (\ref{a.4}), it can be established \cite{BN08}, \cite[Appendix A]{AL22} that for $t < 4$,
  \begin{equation}\label{E.1}
    \rho_{(1), \infty}(L_0(t) -x;t) \mathop{\sim}\limits_{x \to 0^+}  A(t)  \sqrt{x}, \qquad A(t) = {1 \over \pi} \sqrt{2 \over t^{3/2} (4 - t)^{1/2}},
 \end{equation}      
and that for $t=4$,
 $  \rho_{(1), \infty}( \pi -x;t) \asymp  |x|^{1/3}$ as $x \to 0$. The former is well known in
 random matrix theory as a characteristic of a soft edge  known from the boundaries of the Wigner semi-circle law
 --- see e.g.~\cite[\S 1.4]{Fo10} --- while the cusp $|x|^{1/3}$ characterises the
 Pearcey singularity \cite{BH98,HHN16,EKS20}. It is also known that for small $t$ the functional form of the density is well
 approximated by a  Wigner semi-circle functional form, which becomes exact in a scaling limit \cite{Ka22}.
 
 \begin{remark}\label{R4.1a}
  One of the results obtained in \cite{Ka22} in relation to the unitary Hermite polynomials (\ref{3.28}) concerns
  their zeros. In particular,
 it is shown that all the zeros of (\ref{3.28}) are on the unit circle in the complex $z$-plane, and with $z=e^{ix}$,
 $q=e^{-t/2n}$, their density for $n \to\infty$ is given by $ \rho_{(1), \infty}(x;t) $. Our expression for the $\tilde{H}_n(z;q)$
 as an average with respect to the PDF (\ref{3.0c5}) obtained in Remark \ref{R3.2}.2 gives a different viewpoint on the
 result for the density.
  Thus for the average one has the limit formula (see e.g.~\cite[Lemma 3.3]{FW17} for justification)
 \begin{equation}\label{3.29} 
  \lim_{n \to \infty} {1 \over n} \log  \Big  \langle \prod_{k=1}^n \sin ( ( x - x_k)/2)  \Big \rangle \Big |_{q=e^{-t/2n}} = {1 \over 2 \pi} \int_I \log [\sin ((x-y)/2)] 
  \rho_{(1), \infty}(y;t) \, \mathrm dy, \quad x \notin I,
 \end{equation}  
where $I \subset [-\pi,\pi]$ is the interval of support of $  \rho_{(1), \infty}$, which interpreted in terms of $\tilde{H}_n(z;q)$
implies the limiting density of zeros result from \cite{Ka22}.
  \end{remark}
 
 \subsection{Schur function average}
 Let $z_j = e^{ i x_j }$, and let $\langle \cdot \rangle$ denote an average with respect the PDF
 (\ref{3.0c5}), which is dependent of $q = e^{-t/2N}$.   The moments of the spectral density then correspond to computing 
 \begin{equation}\label{A.1}
m_k^{(N)}(t) = {1 \over N}  \Big \langle \sum_{j=1}^N z_j^k  \Big \rangle.
  \end{equation}  
  In the theory of symmetric polynomials, $\sum_{j=1}^N z_j^k $ is referred to as the power sum, which for $N$
  large enough can be used to form a basis \cite{Ma95}. Another prominent choice of basis for symmetric polynomials
  is the Schur polynomials $\{S_\kappa \}$,
  labelled by a partition $\kappa = (\kappa_1, \dots, \kappa_N)$, where $\kappa_1 \ge \kappa_2 \ge \cdots \ge \kappa_N \ge 0$.
  They can be defined in terms of a determinant
  according to
   \begin{equation}\label{A.2}
   S_\kappa(z_1,\dots,z_N)   =  { \det [ z_k^{N - j + \kappa_j} ]_{j,k=1}^N \over \Delta_N(- \mathbf z)},
  \end{equation} 
  with $\Delta_N( \mathbf z)$ specified according to (\ref{5.1c}).
 In fact one has the identity (see e.g.~\cite{Ma95})
  \begin{equation}\label{A.3} 
 \sum_{j=1}^N z_j^k = \sum_{r=0}^{{\rm min} \, (k-1,N-1)} (-1)^r S_{(k-r,1^r)} (z_1,\dots, z_N),
 \end{equation}
 where $(k-r,1^r)$ denotes the partition with largest part $\kappa_1 = k  - r$, $r$ parts $(r \le N - 1)$ equal to
 1 and the remaining parts equal to 0. Thus
   \begin{equation}\label{A.4}  
 \Big   \langle \sum_{j=1}^N z_j^k  \Big \rangle   =  \sum_{r=0}^{{\rm min} \, (k-1,N-1)} (-1)^r  \langle
 S_{(k-r,1^r)} (z_1,\dots, z_N) \rangle .
  \end{equation}  
  Evaluating the RHS of (\ref{A.4}) is the method sketched in \cite{On81,AO84} to obtain (\ref{4.0c}) ---
  details of the calculation were not given. Here we provide the details, showing too that the
  crucial step of  evaluating  the Schur polynomial average can be carried out within the
  general formalism based on cyclic P\'olya ensembles.

  On this, define the complex form of $p_t(\mathbf x;\mathbf 0)$ as implied by (\ref{5.1b}), and denote this by $f_U$.
 Let $s_k = N - k + \kappa_k$ $(k=1,\dots, N)$.
  Recalling the definition of the spherical transform (\ref{2.6}), we see that
    \begin{equation}\label{A.5}  
    \mathcal S f_U = {\prod_{j=0}^{N-1} j! \over \Delta_N(- \mathbf s) |_{s_k = N - k + \kappa_k}} 
    \langle   S_\kappa(z_1,\dots,z_N) \rangle.
  \end{equation}  
  Hence knowledge of $  \mathcal S f_U$, which is known for general    cyclic P\'olya ensembles from
  (\ref{sf_cpe}), suffices for the calculation of the Schur polynomial average.
  
  \begin{proposition}
  In relation to the complex form of $p_t(\mathbf x;\mathbf 0)$ as implied by (\ref{5.1b}), we have
   \begin{equation}\label{A.6}  
    \langle    S_\kappa(z_1,\dots,z_N)   \rangle   =   \prod_{1 \le j < k\le N} { k - j + \kappa_j - \kappa_k \over
    k - j } \prod_{j=1}^N q^{\kappa_j^2 + (N -2j + 1) \kappa_j}.
   \end{equation}  
   In particular
    \begin{equation}\label{A.7}  
    \langle   S_{(k-r,1^r)} (z_1,\dots, z_N) \rangle =   q^{k^2 + k (N - 1) -2kr} {(N-1+k)! \over (N-1)! k!} {(-1)^r \over r!} {(-k+1)_r (-N+1)_r \over (-(k-1+N))_r},
     \end{equation} 
     where $(u)_r := u(u+1) \cdots (u+r-1)$ denotes the increasing Pochhammer symbol.    
   \end{proposition}   
   
   \begin{proof}
   The result (\ref{A.6})  follows immediately from (\ref{A.5}), upon noting
   $$
    \Delta_N(- \mathbf s) |_{s_k = N - k + \kappa_k} = \prod_{1 \le j < k\le N} (k - j + \kappa_j - \kappa_k)
    $$
    and noting from  (\ref{sf_cpe}) and (\ref{sf_cpe+}) that
    $$
    \prod_{j=1}^N {\mathcal S (N - j + \kappa_j) \over \mathcal S(N-j)} = \prod_{j=1}^N q^{\kappa_j^2 + (N -2j + 1) \kappa_j}.
    $$
    
    With $\kappa = (k-r,1^r)$, a direct calculation shows
    $$
     \prod_{j=1}^N q^{\kappa_j^2 + (N -2j + 1) \kappa_j} = q^{k^2 + k (N - 1) -2kr},
     $$
     and furthermore
     $$
  \prod_{1 \le j < k\le N} { k - j + \kappa_j - \kappa_k \over
    k - j } = A_1 A_2 A_3,
    $$
    where
    \begin{align*}
    A_1 & := \prod_{q=2}^{r+1} {q - 2 + k - r \over q - 1} = {(-1)^r \over r!} (-k+1)_r, \\
    A_2 & := \prod_{p=2}^{r+1} \prod_{q=r+2}^N {q-p+1 \over q - p} =     {(-1)^r \over r!} (-N+1)_r, \\
    A_3 & := \prod_{q=r+2}^N {q - 1 + k - r \over q - 1} = (-1)^r  {r! \over (N-1)!} {(N-1+k)! \over k!} {1 \over (-(k-1+N))_r}.
    \end{align*}
    The result (\ref{A.7}) now follows.

     \end{proof}
     
     \begin{coro}
     With $q = e^{-t/2N}$ we have
   \begin{align}\label{A.8}     
  m_k^{(N)}(t)  & =      q^{k^2 + k (N - 1) } {(N-1+k)! \over k! N! } 
  \, {}_2 F_1(1-N,1-k ;-(k-1+N); q^{-2k})  \nonumber \\
 & =   q^{k^2 + k (N - 1) } \, {}_2 F_1(1-N,1-k ; 2 ; 1 - q^{-2k}). 
   \end{align}  
   \end{coro}
   
   \begin{proof}
   The series expansion of the Gauss hypergeometric function is
    \begin{equation}\label{A.9}     
     \, {}_2 F_1(a,b;c;z) = \sum_{n=0}^\infty {(a)_n (b)_n \over n! (c)_n} z^n.
    \end{equation} 
    With this noted, the first equality follows by substituting (\ref{A.7}) in (\ref{A.4}) then
    substituting the result in (\ref{A.1}).
    To obtain the second equality, we make use of the polynomial identity
    \begin{equation}\label{A.10}       
    \, {}_2 F_1(-a,-b;c;z) =   {(c+a+b-1) \cdots (c + b) \over (a+c-1) \cdots c} \, {}_2 F_1(-a,-b;-a-b+1-c;1-z),
   \end{equation} 
   valid for $a \in \mathbb Z_{\ge 0}$.  This can be checked by verifying that both sides
   satisfy the Gauss hypergeometric function differential equation, and agree at $z=0$. In fact (\ref{A.10})
   is a special case of the connection formula expressing a solution analytic about $z=1$ as a linear
   combination of the two linearly independent solutions about $z=0$; see e.g.~\cite[Th.~2.3.2]{AAR99}.
    \end{proof}    
    
    \begin{remark} $ $\\
    1.~Applying one of the standard Pfaff transformations to the second of the hypergeometric formulas in
    (\ref{A.8}) gives a third form
     \begin{equation}\label{A.8a}
      m_k^{(N)}(t)   = q^{k^2 - k (N - 1) } \, {}_2 F_1(1-N,1+k ; 2 ; 1 - q^{2k}). 
    \end{equation} 
    All of the three forms can be equivalently expressed in terms of Jacobi polynomials $P_{N-1}^{(a,b)}(z)$
    for certain parameters $(a,b)$ and argument $z$.  For future reference,
    we make explicit note of the Jacobi polynomial form equivalent to (\ref{A.8a}),
     \begin{align}\label{A.8b} 
       m_k^{(N)}(t)    =    {1 \over N}  q^{k^2 - k (N - 1) }P_{N-1}^{(1,k-N)}(2q^{2k}-1) 
        = {(-1)^{N-1}  \over N}   q^{k^2 - k (N - 1) }  P_{N-1}^{(k-N,1)}(1-2q^{2k}) .
     \end{align}       
     2.~As remarked in Section \ref{S1.3}, the distribution of the ensemble of matrices obtained from Dyson Brownian motion on $U(N)$ at a given
    time is unchanged by conjugation with a fixed $U \in U(N)$, telling us that this ensemble exhibits unitary symmetry.
    The classical Gaussian, Laguerre, Jacobi and Cauchy unitary ensembles all share the property seen in (\ref{A.8})
    that their moments can be expressed in terms of hypergeometric polynomials \cite{WF14,CMOS19,
    ABGS21,FR21}. However, these classical ensembles have a property not shared by (\ref{A.8}), whereby the sequence of
    moments satisfy a three term recurrence. Hypergeometric polynomials and their basic $q$ extensions
    have also been shown to provide the closed form expression for the moments of certain discrete and
    $q$ generalisations of ensembles of random matrices with unitary symmetry \cite{CCO20,Fo22,Co21,FLSY21}.
    \end{remark}

    \subsection{Moment evaluation from the cyclic P\'olya density formula}
     According to (\ref{2.16}) in the case $m=1$ and (\ref{kernel}), the density $\rho_{(1),N}$ for the PDF
    (\ref{3.0c5}) has the explicit functional form
     \begin{equation}\label{kernelA}
\rho_{(1),N}(x;t)	  =  \frac{1}{2\pi} \bigg( N +\sum_{k=0}^{N-1}\sum_{l\in\mathbb Z\backslash\{0,\ldots,N-1\}}
	 \frac{\Gamma(N-l)q^{(l-(N-1)/2)^2-(k-(N-1)/2)^2}}{\Gamma(-l)\Gamma(N-k)\Gamma(k+1)}\frac{(-1)^{k}e^{ ix(k-l)}}{k-l}\bigg).
	 	\end{equation}
This provides an alternative starting point to deduce the first equality in (\ref{A.8}) for
 $m_k^{(N)}(t)$.

   \begin{proposition}
   For $k \ge 1$ we have
    \begin{equation}\label{m1}
   m_k^{(N)}(t) = {1 \over k N} q^{k^2 + (N - 1)k} \sum_{p=0}^{N-1} (-1)^{p} {\Gamma(N  + k - p) \over \Gamma(k-p) \Gamma(N-p) \Gamma(p+1)} q^{-2kp}.
   \end{equation}
   This sum can be identified with the first hypergeometric form in   (\ref{A.8}).
    \end{proposition}
    
    \begin{proof}
    We replace the summation label $k$ in (\ref{kernelA}) by the symbol $p$. Then we write $l=p-k$, where $k$ is an
    independent positive integer. Since $m_k^{(N)}(t)$ is the coefficient of $e^{ i x k }$ in
    $\rho_{(1),N}(x;t)/N$, the formula (\ref{m1}) results.
    Using the formula
    $$
    {\Gamma(u) \over \Gamma(u-p)} = (-1)^p (1-u)_p
    $$
    to rewrite the gamma functions in the summand gives the series form of the hypergeometric polynomial in the first line
    of (\ref{A.8}).
    \end{proof}
    
    \begin{remark}
    The expression on the RHS of (\ref{m1}) was first obtained in a field theory context in \cite[Eq.~(17)]{BGV99}.
    It was related to the result (\ref{4.0c}) of 
 \cite{On81,AO84}  in \cite{GT14}.
 \end{remark}

 \subsection{Asymptotics}\label{S4.3}
 According to the form (\ref{A.8b}) we see that knowledge of the large $k, N$ asymptotic form of $m_k^{(N)}(t)$
  is reliant on knowledge of the large $k,N$ asymptotic form of the Jacobi polynomial
  $P_{N-1}^{(k-N,1)}( 1 - 2 \lambda^2)$, with $0<\lambda <1$. 
  This has been investigated in the literature some time ago \cite{CI91}, but the formulas obtained therein
  contain some inaccuracies \cite{FFN10}.
   Fortunately an accurate and user friendly statement is available in the recent work
  \cite{SZ22}. The statement of the result requires some notation.
  Introduce the  complex number
  $z_+ = z_+(\mu,\lambda)$ by
    \begin{equation}\label{S.20}  
    z_+ = u(\mu,\lambda) + i \sqrt{1 - (u(\mu,\lambda)^2}, \qquad  u(\mu,\lambda) = {(\mu - 1) (1 + \lambda^2) + 2 \lambda^2 \over 2 \lambda \mu},
    \end{equation}
    where it is assumed the parameters $\mu, \lambda$ are such that $|u(\mu,\lambda) | \le 1$ so that $z_+$ has unit modulus and is
    in the upper half plane.
    Now define $\phi_+ = \phi_+(\mu,\lambda) \in [0,\pi]$ as the argument of $z_+$ so that $e^{i \phi_+} = z_+$. In terms of $\phi_+$ define
   \begin{equation}\label{h0}    
    h(\mu,\lambda) = {\rm arg} \Big ( {z^{\mu } (1 - \lambda z ) \over z - \lambda} \Big ) \Big |_{z = e^{i \phi_+} }=
    \mu \phi_+ + {\rm arg} \Big ( {1 - \lambda e^{i \phi_+} \over  e^{i \phi_+} - \lambda} \Big ),
      \end{equation}
      where here we are free to take the argument function as multivalued.

  \begin{proposition}\label{P5.5}
  (\cite[Th.~1]{SZ22})
  Set $\mu = k/N$, and require that for $k,N \to \infty$ we have $0<\mu < 1$. 
  For 
   \begin{equation}\label{S.18a}  
      1 > \lambda >  { 1-\mu \over 1+\mu}, 
   \end{equation}    
  the large $N$ expansion
  \begin{multline}\label{S.19}   
  \lambda^{k-N} P_{N-1}^{(k-N,1)}( 1 - 2 \lambda^2) = {\sqrt{2 \over N \pi}} 
  {((1 - \lambda^2) \mu)^{-1/2} \over ((1 - \lambda^2)((\mu+1)^2\lambda^2 - (\mu - 1)^2)^{1/4}}
  \\
  \times
  \cos (N h(\mu,\lambda) + \pi/4) + {\rm O} \Big ({1 \over N^{3/2}} \Big )
  \end{multline}
  holds true, and moreover holds uniformly in $\lambda$ for $\lambda$ a compact interval in
  $((1 - \mu)/(1+\mu),1)$. (On this latter point see \cite[\S 4.1]{Co05}.)

   In contrast, for $0 \le  \lambda <   (1-\mu)/(1+\mu)$, the LHS of (\ref{S.19}) 
  decays exponentially fast in $N$.
  
   Suppose instead that $\mu > 1$. 
 For
  \begin{equation}\label{S.23}   
  0 < \lambda < {\mu - 1 \over \mu + 1}
 \end{equation}  
 the asymptotic formula (\ref{S.19}) again remains valid, while for 
 \begin{equation}\label{S.24}   
  {\mu - 1 \over \mu + 1} < \lambda < 1
 \end{equation}  
 the LHS of (\ref{S.19}) decays exponentially fast in $N$.      
  \end{proposition}
  
  We are now in a position to establish the sought asymptotic forms of $m_k^{(N)}(t)$.
  
  \subsection*{Proof of Corollary \ref{C1a}}
  With $\lambda = e^{- \mu t/2}$, the inequalities (\ref{S.18a}) and (\ref{S.23}) can be recast in terms of $t$
  and $\mu$ to become equivalent to the requirement that $0 < t < t^*$ where $t^*$ is given by (\ref{c.1}). In these
  cases we substitute (\ref{S.19}) for the Jacobi polynomial in the second expression of (\ref{A.8b}), to obtain
  (\ref{c.2}). Outside of this interval, Proposition \ref{P5.5} tells us that the Jacobi polynomial, and thus the moments,
  decay exponentially fast. \hfill $\square$
  
  \subsection*{Proof of Corollary \ref{C1b}} 
  Inspection of the working of \cite{SZ22} for the derivation of (\ref{S.19}) shows that the leading term,
  appropriately expanded, remains valid in
  an extended regime with $k,N \to \infty$ and $k \ll N$, and is thus uniform for $\mu \to 0^+$
  provided $N \mu \to \infty$. The keys are the validity of the inequality (\ref{S.18a}),
  and that $u(\mu,\lambda)$ in (\ref{S.20}) remains well defined. On the former point, substitute $\lambda = e^{- t \mu/2}$ and
  expand for small $\mu$. We see that (\ref{S.18a}) is valid provided $t < 4$. On the latter point, doing the same in
  (\ref{S.20})  shows
  \begin{equation}\label{uu}    
 u(\mu,\lambda) \to 1 - t/2,
  \end{equation}  
 and thus $z_+ = (1 - t/2) + i \sqrt{t(1 - t/4)}$. The significance of $z_+$ in the working of \cite{SZ22} is that it is the saddle point
 in the method of stationary phase
 
 Again setting $e^{i \phi_+} = z_+$, we thus have
 $\cos \phi_+ = 1 - t/2$, $\sin \phi_+ = \sqrt{t(1 - t/4)}$, which in turn allows us to compute from
 (\ref{h0}) the small $\mu$ expansion
  \begin{equation}\label{h0+}     
  h(\mu,\lambda) |_{\lambda = e^{-\mu t/2}} \sim  \pi +\mu L_0(t) + {\rm O}(\mu^2),
 \end{equation}  
 where $L_0(t)$ is specified by (\ref{4.0e}). Also, for small $\mu$, 
  \begin{equation}\label{h0+1} 
  {((1 - \lambda^2) \mu)^{-1/2} \over ((1 - \lambda^2)((\mu+1)^2\lambda^2 - (\mu - 1)^2)^{1/4}} \Big |_{\lambda = e^{-\mu t/2}} \sim {1 \over \mu^{3/2}}  \sqrt{1 \over t^{3/2} (4 - t)^{1/2}},
  \end{equation}  
  Recalling (\ref{A.8b}) allows us to conclude from (\ref{S.19}) with $\lambda = e^{-\mu t/2}$, and expanded for small $\mu$, the leading
  $k,N \to \infty$ and $k \ll N$ expansion of $ m_k^{(N)}(t)$ as specified by (\ref{h0+2}).
   \hfill $\square$
     
     \begin{remark}\label{R4.8} 
  The asymptotics of the Jacobi polynomial in
  Proposition \ref{P5.5} is also known in the cases $ \lambda = | (\mu - 1)/ (\mu + 1)|$, with the leading order $N$
  dependence then proportional to $N^{-1/3}$ \cite{SZ22}. Hence for $t=t^*$ the moments $ m_k^{(N)}(t)$ decay
  as $N^{-4/3}$. Note that the exponent $-4/3$ is the same as that for the asymptotic decay of $m_k^{(\infty)}(t)$
 at $t=4$, as noted above (\ref{4.0e}) from \cite{Le08}, as we would expect. Furthermore, repeating the reasoning of
 the first sentence in the paragraph below (\ref{Ma8}) gives $k \asymp N^{6/11}$.
  \end{remark}

\section{The spectral form factor}\label{S5}
 
\subsection{Evaluation of $S_N(k;t)$} 
According to (\ref{S.3}) and (\ref{2.16a}) we have
\begin{align}\label{S.6}
S_N(k;t) 
 = N -  \int_{-\pi}^{\pi} \mathrm dx   \int_{-\pi}^{\pi} \mathrm dy \,  e^{ i k ( x - y)} K_N(x,y) K_N(y,x).
\end{align} 
 The exact expression (\ref{kernel}) allows this
double integral to be evaluated in terms of a sum.

  \begin{proposition}
  For $j,l$ integers, define
  \begin{equation}\label{S.7}
  a_j = {(-1)^j q^{-(j-(N-1)/2)^2} \over \Gamma(N-j) \Gamma(j+1)}, \quad b_l = {\Gamma(N - l) q^{(l-(N-1)/2)^2} \over \Gamma(-l)}.
  \end{equation}
  For these quantities to be nonzero we require $0 \le j \le N-1$ and $ l \notin \{0,\dots,N-1\}$ respectively.
  For $k$ a non-negative integer, we have
   \begin{equation}\label{S.7a}
   S_N(k;t) = \min (k,N) + \sum_{j=0}^{\min (k-1,N)} \sum_{l=\max (k,N)}^{N+k-1} {a_j b_{j-k} a_{l-k} b_l \over (j-l)^2}.
   \end{equation}
  \end{proposition}   
  
  \begin{proof}
  In terms of the notation (\ref{S.7}), the correlation kernel (\ref{kernel}) can be written as
   \begin{equation}\label{S.8} 
  2 \pi  K_N(x,y)=  \sum_{l=0}^{N-1} e^{ il(x-y)}+ 
  \sum_{j=0}^{N-1}\sum_{l\in\mathbb Z\backslash\{0,\ldots,N-1\}} {a_j b_l \over j - l} e^{ i(xj-yl)}.
  \end{equation}	
 For $s$ a non-negative integer and $f(x,y)$ is a $2 \pi$-periodic function
 in both $x$ and $y$,
 introduce the notation $[e^{ i s(x-y)}]\, f(x,y)$
 to denote the coefficient of $e^{ i s(x-y)}$ in the corresponding double Fourier series. It follows from
 (\ref{S.8}) that
  \begin{equation}\label{S.9} 
 (2 \pi)^2  [e^{ i s(x-y)}]\, K_N(x,y) K_N(y,x) = \max (N-s,0) -  \sum_{j=0}^{\min (s-1,N)} \sum_{l=\max (s,N)}^{N+s-1} {a_j b_{j-s} a_{l-s} b_l \over (j-l)^2}.
   \end{equation}
   Using this result in (\ref{S.6}), and noting too from this that $S_N(k;t) = S_N(-k;t)$, (\ref{S.7a}) results.
\end{proof}

For $k \le N$ in (\ref{S.7a}) we have the rewrite
\begin{equation}\label{S.10}
   S_N(k;t) =  k + \sum_{j=0}^{k-1} \sum_{l'=0}^{k-1} {a_j b_{j-k} a_{N-1-l'} b_{N+k-1-l'} \over (j+l'-N-k+1)^2}.
   \end{equation}
   Substituting according to (\ref{S.7}) and simplifying then gives
\begin{multline}\label{S.11}  
   S_N(k;t) =  k - q^{2k^2 + 2k (N - 1)} \sum_{j=0}^{k-1} \sum_{l'=0}^{k-1}  q^{-2k(j +l')} (-1)^{j+l'} \\
   \times {\Gamma(N+k-j) \Gamma(N+k-l') \over \Gamma(N-j) \Gamma(j+1) \Gamma(l'+1) \Gamma(N-l') \Gamma(k-j) \Gamma(k-l')} {1 \over  (j+l'-N-k+1)^2}.
   \end{multline}
   In the case $k > N$, this expression again holds with the modifications that the first $k$ on the RHS is to be replaced by $N$, as is each $k$
   in the upper terminals of the sums.
   The fact that at $t=0$ the initial matrix is the identity, and so all eigenvalues are unity, implies from the definition (\ref{S.4}) that 
   \begin{equation}\label{S.11a}
   S_N(k;t) \Big |_{t=0} = 0;
    \end{equation}
   since $q=1$ for $t=0$, note that according to (\ref{S.11}) this implies an identity for the double sum.

   For $k=1$ the double sum in (\ref{S.11}) consists of a single term, which  after simplification and substituting for $q$ according to (\ref{3.0c3}) reads
  \begin{equation}\label{S.12} 
  S_N(k;t)  \Big |_{k=1} = 1 - e^{-t}.
     \end{equation}
     The next simplest case is $k=2$, when (\ref{S.11}) consists of four terms. However the summand is symmetric in $j,l'$ so this can
     immediately be reduced to three terms,
  \begin{equation}\label{S.13}     
  S_N(k;t)  \Big |_{k=2} =2 - e^{-2t}  \Big ( N^2 e^{-2t/N} - 2 (N^2 - 1) + N^2 e^{2t/N} \Big ).
  \end{equation}  
  Expanding for large $N$ gives 
    \begin{equation}\label{S.14}  
    S_N(k;t)  \Big |_{k=2} =2 - 2e^{-2t}  \Big (   (1 + 2 t^2) + {\rm O} \Big ( {1 \over N^2} \Big ) \Big ).
  \end{equation}  
  
  \begin{remark}
  For Dyson Brownian motion from the identity in the case of $SU(N)$ as distinct from $U(N)$ (in $SU(N)$ the determinant is constrained to be unity)
   a formula similar in appearance to (\ref{S.11}), derived using a Schur polynomial/ heat kernel expansion of the corresponding eigenvalue PDF,
   can be found in \cite[Th.10.2]{LM10}.
   \end{remark}

   \subsection{Large $N$ limit of $S_N(k;t)$ for fixed $k$}
     For general $k$ and fixed $q$, the summand in (\ref{S.11}) appears to be a polynomial in $N$ of degree $2(k-1)$. Yet we expect the
     large $N$ limit of $S_N(k;t)$ to be well defined. A similar feature is already present in the first form of $m_k^{(N)}(t)$ from
   (\ref{A.8}), or alternatively in (\ref{m1}), in which the summand is a polynomial in $N$ of degree $k$.
    It is only after application of the  identity (\ref{A.10}) that the second form in  (\ref{A.8}) is
   obtained, which allows for the computation of the large $N$ limit (\ref{4.0b}). In fact (\ref{S.11}) can be written in terms of an integral of 
   two ${}_2 F_1$ polynomials of the same type as appearing in the first line of  (\ref{A.8}). After applying the transformation (\ref{A.10}) an explicit
   expression for the $N \to \infty$ limit can be obtained.
   
   \begin{proposition}\label{P5.2}
   The double sum formula   (\ref{S.11})  for $S_N(k;t)$ permits the integral forms
   \begin{align}\label{S.15}  
   S_N(k;t) & =  \min (k,N) - q^{2k^2 + 2k (N - 1)}   \bigg ( {(N-1+k)! \over (k-1)! (N-1)!}  \bigg )^2  \nonumber \\
   & \qquad \times \int_0^\infty s e^{-s(N+k-1)} \Big (
   {}_2 F_1(1 - N, 1 - k; -(k-1+N); q^{-2k} e^s) \Big )^2 \,\mathrm ds \nonumber \\
   & =  \min (k,N) - q^{2k^2 + 2k (N - 1)} (kN)^2 \int_0^\infty s e^{-s(N+k-1)} \Big (
   {}_2 F_1(1 - N, 1 - k; 2; 1 - q^{-2k} e^s) \Big )^2 \, \mathrm ds.
   \end{align}
   Consequently, for $k$ fixed with respect to $N$,
     \begin{align}\label{S.16}  
   S_\infty(k;t)  : = \lim_{N \to \infty}   S_N(k;t)  & = k - e^{-tk} k^2 \int_0^\infty s e^{-s} \Big (
   {}_1 F_1(1 - k; 2; kt + s) \Big )^2 \, \mathrm ds \nonumber \\
& =  k - e^{-tk}  \int_0^\infty s e^{-s} \Big ( L_{k-1}^{(1)}(kt+s) \Big )^2
 \, \mathrm ds.
   \end{align}
    \end{proposition}   
  
  \begin{proof}
  With $u = N+k-1-j-l'$ we rewrite the factor of $1/u^2$ in  (\ref{S.11}) according to the simple integral
  $$
  {1 \over u^2} = \int_0^\infty s e^{-us} \, \mathrm ds.
  $$
  Taking the integral outside of the double summation, the latter factorises into the product of
  two single summations, both of which have an identical structure to that in (\ref{m1}), which
  we know can be identified with the first line of  (\ref{A.8}). This gives the first expression in
  (\ref{S.15}). The second follows by applying the  transformation (\ref{A.10}), as is analogous with the second line of (\ref{A.8}). The limit (\ref{S.16})
  now follows by changing variables $s \mapsto s/N$ in this  latter expression,
  and proceeding as in deducing (\ref{4.0b}) from (\ref{4.0c}). An additional technical issue is the need to interchange the order of the
  limit with the integral. Here one first note that the hypergeometric polynomial in the final line of (\ref{S.15}) is, for fixed $k$ and $N$ large enough, of degree $2(k-1)$
  in $(1 - q^{2k} e^s)$. It follows that the task of justifying the interchange of the order can be reduced to justifying this for
  \begin{multline*}
  \lim_{N \to \infty} \int_0^\infty s e^{- s (1 + (k-1)/N)} (-N)^p (1 - q^{-2k} e^{s/N})^p \, ds \\
  =   \lim_{N \to \infty} \int_0^\infty s e^{- s (1 + (k-1)/N)}  N^p (1 - e^{(s+k)/N})^p \, ds,
  \end{multline*}
  with $p \le 2 (k-1)$, $p \in \mathbb Z^+$. For $N$ large enough, $e^{-s(k-1-p)/N} \le e^{s/2}$. Also, the elementary inequality
  $1 - e^{-x} \le x$ shows that $N^p (1 - e^{-(s+k)/N})^p$ is bounded by $(s+k)^p$. Hence for $N$ large enough the integrand is bounded pointwise
  by $s e^{-s/2} (s+k)^p$, which is integrable on $s \in \mathbb R^+$. The justification of the interchange of the order now follows by dominated convergence.
  
  \end{proof}
  
  \begin{remark}
  The Laguerre polynomials permit the addition formula
  \begin{equation}\label{S.17}    
  L_p^{(\alpha + \beta + 1)}(x+y) = \sum_{s=0}^p L_s^{(\alpha)}(x)  L_{p-s}^{(\beta)}(y).
  \end{equation}
  Applying this relation to each $ L_{k-1}^{(1)}(kt+s)$ in (\ref{S.16}) with $p=k-1, x = s, y = kt, \alpha = 1, \beta = -1$,
  and using the orthogonality
  $$
  \int_0^\infty s e^{-s} L_p^{(1)}(s)  L_q^{(1)}(s) \, \mathrm ds = (p+1) \delta_{p,q}
  $$
  shows that the integral in (\ref{S.16}) permits an evaluation as the finite sum (\ref{S.18x}).
 \end{remark} 
     
  \subsection{Large $N$ limit of $S_N(k;t)$ for  $k$ proportional to $N$ --- proof of Theorem \ref{P5.6}}   
  We first consider the case $k \le N$. Changing variables $s \mapsto sk/N$ in (\ref{S.15}) and using the equality between  (\ref{A.8a}) and
  (\ref{A.8b}) shows
  \begin{equation}\label{S.18}      
 S_N(k;t) =  k  - {k^4 \over N^2} \int_0^\infty s e^{-(s+t) k ( k - N + 1)/N} \Big ( P_{N-1}^{(k-N,1)}(1 - 2 e^{- (s+t)k/N}) \Big )^2 \, \mathrm ds.
  \end{equation}
  We see that the asymptotic formulas for
  $P_{N-1}^{(k-N,1)}( 1 - 2 \lambda^2)$, with $0<\lambda <1$, given in Proposition \ref{P5.5} are again relevant.
  It is application of these formulas which allows
  the sought scaling limit of (\ref{S.18}), as specified in  Theorem \ref{P5.6}, to be computed.
  
    Suppose first that $0 < t < t^*$. In relation to Proposition \ref{P5.5}, set $\lambda = e^{- \mu (s+t)/2}$.
    Taking into consideration the range of $\lambda$ for which (\ref{S.19}) holds uniformly,
    we see that for $0 < s < t^* - t - \epsilon$, $0 < \epsilon \ll 1$, the factor of the integrand
      \begin{equation}\label{S.22}  
    ( \lambda^{k-N} P_{N-1}^{(k-N,1)}( 1 - 2 \lambda^2))^2
      \end{equation}   
     can be replaced by the square of the leading term
    on the RHS of (\ref{S.19}). Making use too of (\ref{c.1}) shows
    \begin{multline}\label{S.22R}
    {k^4 \over N^2} \int_0^{ t^* - t - \epsilon} s e^{-(s+t) k ( k - N + 1)/N} \Big ( P_{N-1}^{(k-N,1)}(1 - 2 e^{- (s+t)k/N}) \Big )^2 \, \mathrm ds \\
    \mathop{\sim}\limits_{N \to \infty} {2 \over \pi}  {\mu^3 \over (1 + \mu) } 
     \int_0^{ t^* - t - \epsilon} s f(\mu;\lambda) \cos^2(N \tilde{h}(t+s,\mu) +  \pi/4) \Big |_{\lambda = e^{- \mu (s+t)/2}}  \, \mathrm ds,
    \end{multline}
    where 
    $$
    f(\mu;\lambda) = {\lambda^2 \over (1 - \lambda^2)^{3/2}} {1 \over (\lambda^2 - e^{- \mu t^*} )^{1/2}}.
    $$
    
    Next we make use of the identity $\cos^2(N \tilde{h}(t+s,\mu) +  \pi/4) = {1 \over 2}
    ( 1 + \cos 2 (N  \tilde{h}(t+s,\mu) +  \pi/4 ))$. It shows that to leading order the
    square of the cosine can  be replaced by ${1 \over 2}$. Thus the term involving $ \cos 2 (N  \tilde{h}(t+s,\mu) +  \pi/4 ))$,
    being multiplied by an absolutely integrable function of $s$ --- specifically $s  f(\mu;\lambda) $ --- must decay in $N$ in accordance with the Riemann-Lebesgue lemma
   for Fourier integrals and its generalisations \cite{Er55}, and hence vanish in the limit.
   We therefore conclude that the $N \to \infty$ limit of (\ref{S.22R}) is equal to 
    \begin{equation}\label{S.22Rs}
{1 \over \pi}   {\mu^3 \over (1  + \mu )  }   \int_0^{ t^* - t - \epsilon} s f(\mu;\lambda) \Big |_{\lambda = e^{- \mu (s+t)/2}}  \, \mathrm ds.
    \end{equation}
   Taking $\epsilon \to 0$
    gives  the term involving the integral
    in (\ref{S.20a}). 
    
    The range $s > t^* - t + \epsilon$ does not contribute due to the corresponding exponential decay in $N$
    of  (\ref{S.22}), as known from the second statement of
   Proposition \ref{P5.5}.  Thus
   \begin{multline}\label{S.22R+}
    {k^4 \over N^2} \int_{ t^* - t + \epsilon}^\infty s e^{-(s+t) k ( k - N + 1)/N} \Big ( P_{N-1}^{(k-N,1)}(1 - 2 e^{- (s+t)k/N}) \Big )^2 \, \mathrm ds \\
  \le  \mu^4   N^2 e^{-N u(\mu,s)}
     \int_{ t^* - t  +  \epsilon}^\infty  s e^{- \mu (s+t)}  \, \mathrm ds
    \end{multline}
    for some $ u(\mu,s) > 0$. Since the RHS tends to zero as $N \to \infty$, so must the LHS.
    The remaining task then is
    to check that there is no separate contribution from the singularity
   of (\ref{S.19}) at $s=t^* - t$.  
   
  Setting $t=0$ to begin, one approach is to verify that  the RHS of (\ref{S.20a}) is already identically zero, 
   as required by the sum rule (\ref{S.11a}).
   This follows from the integration formula
    \begin{equation}\label{S.22a}  
    \int_0^T {s e^{-s} \over (1 - e^{-s})^{3/2}} {1 \over \sqrt{e^{-s} - e^{- T}}} \, \mathrm ds = \pi \Big ( 1 + \tanh (T/4) \Big ),
    \end{equation}  
  as can be validated by computer algebra, used in conjugation with (\ref{c.1}).  As a consequence we  that with $t=0$ in the integral of (\ref{S.18}), there is no
  contribution to its $N \to \infty$ form for $s > t^* - \epsilon$ in the limit $\epsilon \to 0^+$, and thus
    \begin{equation}\label{R1}
    \lim_{\epsilon \to 0^+} \lim_{N \to \infty} \int_{t^* - \epsilon}^\infty 
     s e^{-s \mu ( k - N + 1)} \Big ( P_{N-1}^{(k-N,1)}(1 - 2 e^{- s \mu }) \Big )^2 \, \mathrm ds = 0.
     \end{equation}
     We claim now that knowledge of (\ref{R1}) is sufficient to establish that for general $0 < t < t^*$ fixed in
    (\ref{S.18}), there is no
  contribution to its $N \to \infty$ form for $s > t^* - t - \epsilon$   in the limit $\epsilon \to 0^+$. Note that this is a stronger
  statement than establishing that there is no separate contribution from the neighbourhood of $s=t^* - t$, and in particular
  supersedes the need for (\ref{S.22R+}). The reason for its validity stems from the fact that the
  more general case differs from the case $t=0$ only by the replacement of the factor of the first factor of $s$ in the integrand by
  $(s-t)$, as implied by the explicit functional form
  in (\ref{S.18}), with the $s$ dependence of all the other terms left unaltered and the terminals of integration $(t^*-\epsilon, \infty)$.
  Moreover, since we are assuming $t < t^*$ and we have $s \in (t^*- \epsilon , \infty)$,
  with $\epsilon$ small enough it follows $0 < (s - t) \le s$.  Hence the limit formula (\ref{R1}) remains
  valid with  $s$ in the integrand replaced by
  $(s-t)$, as required for general $0 < t < t^*$, and the claim is established.
     

  The above working involving (\ref{S.22R}) assumed $0 < t < t^*$.  
   For $t > t^*$, the factor   of the integrand (\ref{S.22}) decays exponentially fast in $N$ for all $s > 0$.
   The inequality (\ref{S.22R+}) with the lower terminal of integration replaced by 0 on both sides
   implies 
   the integral in (\ref{S.18}) does not contribute to the $N \to \infty$ limit, which implies (\ref{S.20a}) for this range of $t$.

 The case $k > N$ remains to be considered. According to (\ref{S.15}) and 
  (\ref{A.8b})  the formula (\ref{S.18}) remains valid but with the first $k$ on the RHS replaced by $N$.
  The quantity $k-N$ in $P_{N-1}^{(k-N,1)}$ is now positive, or equivalently $\mu > 1$, which
  is also covered by 
  Proposition \ref{P5.5}. The two cases $0 < t < t^*$ and $ t > t^*$ have to again be considered separately.
  Doing so, and following the strategy of the working  for $\mu < 1$ as detailed above, we arrive at (\ref{S.20})
  in both cases. \hfill $\square$

\begin{figure*}
\centering
\includegraphics[width=0.75\textwidth]{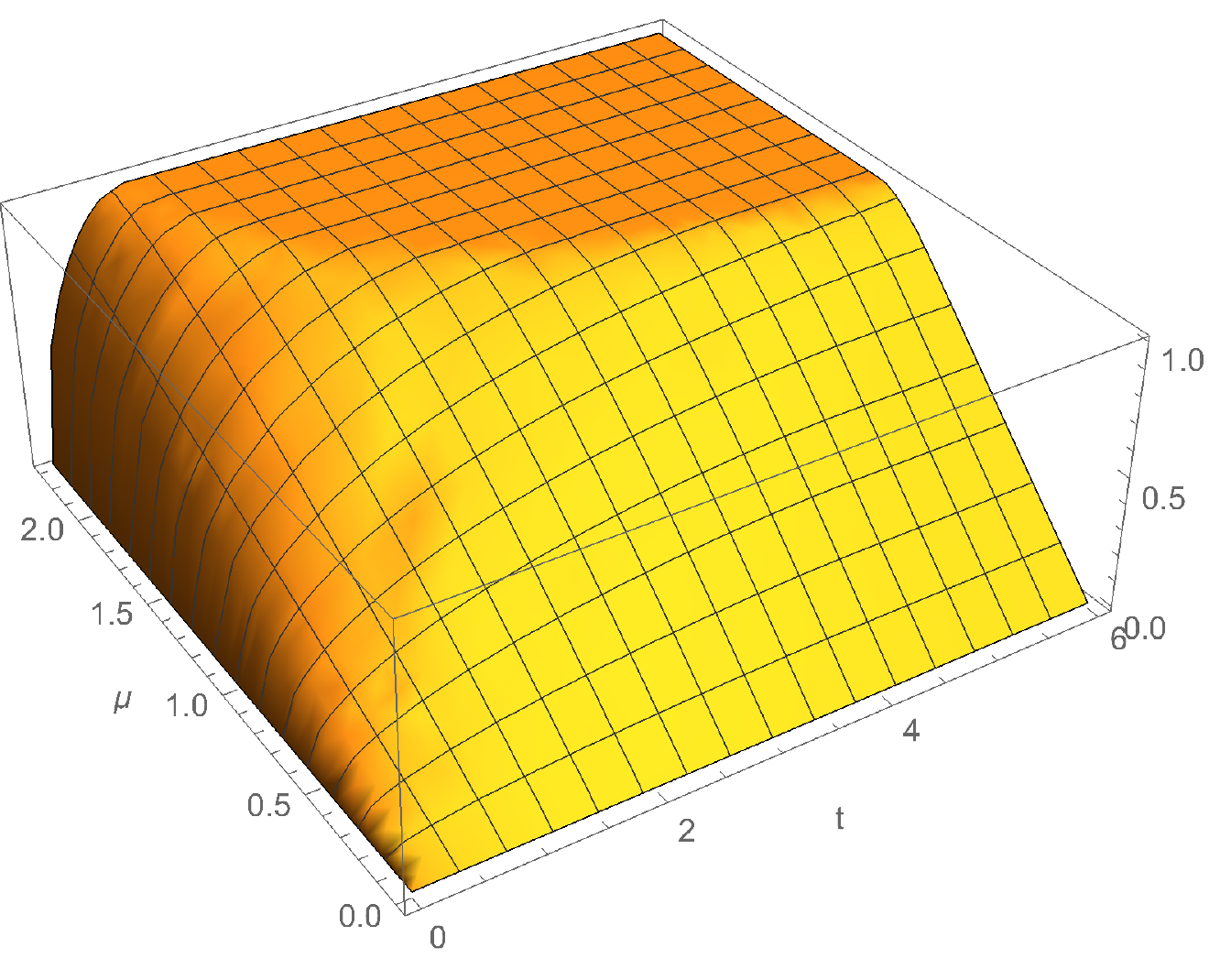}
\caption{Plot of $\tilde{S}_\infty(\mu;t)$ with respect to both $\mu$ and $t$, as specified by Theorem \ref{P5.6}.
The ramp and plateau features are evident in the direction of fixed $t$ and varying $\mu$, with the
transition being smooth in keeping with Remark \ref{R5.7}.4. }
\label{Dfig3}
\end{figure*}

 \begin{remark}\label{R5.7} $ $\\
 1.~As $\mu$ varies from $0$ to $1$, $t^*$ monotonically increases from the value $4$ to $\infty$. In particular,
 this means that for all $t < 4$ there is always a deformation of the limiting CUE result as given by the first term in
 (\ref{S.20a}). On the other hand,
for $ t > 4$ there is always a range of $\mu$ values, $0 \le \mu \le \mu_r$ for which the limiting CUE result is valid.
Here $\mu_r$ is determined  by being the solution of  \ref{c.1}) with $t^*=t$.
 We recall from the discussion of \S \ref{S3.1} that $t=4$ has the significance as the value of $t$
 for which the support of the density becomes the whole unit circle. \\
 2.~As $\mu$ varies from $1$ to $\infty$, $t^*$ monotonically decreases from the value $\infty$ to $0$. 
 This means that for all $t > 0$ there is always a range of $\mu$ values, $\mu_p < \mu < \infty$,  for which the limiting CUE 
 result given by the first term of (\ref{S.20a}) is valid.
 This is in contrast to the circumstance for $0 < \mu < 1$ as described above.
 Both features are in agreement with the heuristic predictions of Section \ref{S1.3}. \\
 3.~ For a given value of $t$, the critical values of $\mu$ introduced above can be read
 off from (\ref{c.1}) to be given by the implicit equations
   \begin{equation}\label{S.22b}  
   t = {2 \over \mu_r} \log {1 + \mu_r \over 1 - \mu_r} \: \: (t > 4), \qquad
    t = {2 \over \mu_p} \log \Big | {1 + \mu_p \over 1 - \mu_p} \Big | \: \: (t >0).
  \end{equation}  
  In keeping with the heuristic discussion of Section \ref{S1.3}, we expect that these equations can be
  rewritten as
     \begin{equation}\label{S.22c}  
     \mu_r = \rho_{(1), \infty}(x;t) \Big |_{x=\pi}  \: \: (t > 4), \qquad \mu_p = \rho_{(1), \infty}(x;t) \Big |_{x=0}    \: \: (t >0),
  \end{equation}  
  with $x=\pi$ ($x=0$) being the points of least (greatest) value of the density. We see from (\ref{a.5}) and
 (\ref{a.3}) with $w = -1$ (for $x=\pi$) and $w = 1$ (for $x= 0 $) that indeed the equations in (\ref{S.22b})
 and (\ref{S.22c}) are equivalent. \\
 4.~Expanding the integral in (\ref{S.20a}) to leading order in $t^* - t$ shows that it is proportional to $(t^* - t)^{3/2}$.
 The exponent $3/2$ is the same as that known for the transition from ramp to plateau in the cases of the
 Gaussian unitary ensemble \cite{BH97}, and the Laguerre unitary ensemble with Laguerre parameter
 scaled with $N$ \cite{Fo21b}. In particular, it implies the function values and first derivatives of $\tilde{S}_\infty(\mu;t)$
 agree across the transition. 
 \end{remark}

 To use Theorem \ref{P5.6} for the computation of $\tilde{S}_\infty(\mu;t)$, with
 $t$ fixed and as a function of $\mu$, for each $\mu$ we must compute $t^*$ as specified by (\ref{c.1}).
 If $0 < t < t^*$ we compute $\tilde{S}_\infty(\mu;t)$ according to the functional form (\ref{S.20}) which involves
 a contribution from the integral, whereas
 for $t \ge t^*$ the term involving the integral is to be set to 0.
 A plot obtained by the  implementation of this procedure, iterated also over $t$, is given in
 Figure \ref{Dfig3}.

\subsection*{Acknowledgements}
	This research is part of the program of study supported
	by the Australian Research Council Discovery Project grant DP210102887.
	S.-H.~Li acknowledges the support of the National Natural Science Foundation of
	China (grant No.~12175155).  A referee is to be thanked for  helpful
	feedback, as too is a handling editor. The support of the chief editor has played
	an important role too.

\nopagebreak

\providecommand{\bysame}{\leavevmode\hbox to3em{\hrulefill}\thinspace}
\providecommand{\MR}{\relax\ifhmode\unskip\space\fi MR }
\providecommand{\MRhref}[2]{%
  \href{http://www.ams.org/mathscinet-getitem?mr=#1}{#2}
}
\providecommand{\href}[2]{#2}


\begin{thebibliography}{10}


\bibitem{AL22}
A.~Adhikari and B.~Landon, \emph{Local law and rigidity for unitary Brownian motion},
arXiv: 2202.06714

\bibitem{AAR99}
G.E. Andrews, R.~Askey, and R.~Roy, \emph{Special functions}, Cambridge
  University Press, New York, 1999.


\bibitem{AO84} 
 G. Andrews and E. Onofri, \emph{Lattice gauge theory, orthogonal polynomials and $q$-hypergeometric functions}, in  Special Functions:  Group Theoretical Aspects and Applications, R.A. Askey, ed.(1984), 163--188.

\bibitem{AKMV10}
P.~Aniello, A.~Kossakowski, G.~Marmo and F.~Ventriglia, \emph{Brownian motion on Lie groups
and open quantum systems}, J. Phys. A  \textbf{43} (2010), 265301.

\bibitem{As19}
 T. Assiotis,  \emph{Intertwinings for general $\beta$-Laguerre and $\beta$-Jacobi processes}, J Theo. Probab.
\textbf{32} (2019), 1880--1891.

\bibitem{ABGS21}
T. Assiotis, B. Bedert, M. Gunes and A. Soor, \emph{Moments of generalized Cauchy random matrices and continuous-Hahn polynomials}, Nonlinearity, \textbf{34} (2021), 4923.

\bibitem{AGS22}
T.~Assiotis, M.A.~Gunes and A.~Soor,  \emph{Convergence and an explicit formula for
the joint moments of the circular Jacobi $\beta$-ensemble characteristic polynomial},
Math. Phys., Analysis and Geometry \textbf{25} (2022), 15.


\bibitem{BGV99}
A. Bassetto, L. Griguolo and F. Vian,  \emph{Instanton contributions to Wilson loops with general winding number
in two-dimensions and the spectral density},  Nucl. Phys. B \textbf{559} (1999), 563--590. 

 \bibitem{Be85}
 M.V.~Berry, \emph{Semiclassical theory of spectral rigidity}, Proc.~R.~Soc.~Lond. A \textbf{400}, (1985) 229--251.


\bibitem{Bi97}
P. Biane,  \emph{Free Brownian motion, free stochastic calculus and random matrices},  Fields Institute Communications \textbf{12},
Amer. Math. Soc. Providence, RI, 1997, 1--19.

\bibitem{BN08}
J.-P. Blaizot and M.A. Nowak,  \emph{Large-$N_c$ confinement and turbulence}, Phys. Rev. Lett., \textbf{101}
(2008), 102001.

\bibitem{Bo98}
A. Borodin, 
\emph{Biorthogonal ensembles.} 
Nucl. Phys. B \textbf{536} (1998), 704--732.

\bibitem{BF22}
P. Bourgade and H. Falconet, \emph{Liouville quantum gravity from random matrix dynamics},
arXiv 2206.03029.

\bibitem{BCL21}
J.~Boursier, D.~Chafa\"i and C.~Labb\'e, \emph{Universal cutoff for Dyson Ornstein Uhlenbeck process},
arXiv:2107.14452

\bibitem{BH97}
E.~Br\'ezin and S. Hikami, \emph{Spectral form factor in a random matrix theory}, Phys. Rev. E
\textbf{55} (1997), 4067--4083.

\bibitem{BH98}
E.~Br\'ezin and S.~Hikami, \emph{Level spacing of random matrices in an
  external source}, Phys. Rev. E \textbf{58} (1998), 7176--7185.
  
     \bibitem{BMS11}
A.~Brini, M.~Mari\~{n}o, and S.~Stevan, \emph{The uses of the refined matrix
  model recursion}, J. Math. Phys. \textbf{52} (2011), 35--51.
  
    \bibitem{CMS17}
A. del Campo, J. Molina-Vilaplana and J. Sonner, \emph{Scrambling the spectral form factor:
unitarity constraints and exact results}, Phys. Rev. D \textbf{95} (2017), 126008. 

\bibitem{CL01}
 E. C\'epa and D. L\'epingle, Brownian particles with electrostatic repulsion on the circle: Dyson?s model for unitary random
matrices revisited, ESAIM Probab. Statist. 5 (2001), 203?224

\bibitem{CI91}
L.-C. Chen and M.E.H. Ismail, \emph{On asymptotics of Jacobi polynomials}, SIAM J. Math. Anal.,
\textbf{22} (1991), 1442--1449.

\bibitem{CL18} X.~Chen and A.W.W.~Ludwig, \emph{Universal spectral correlations in the chaotic
wave function, and the development of quantum chaos}, Phys.~Rev.~B \textbf{98} (2018), 064309.

  
  \bibitem{CMC19}
A.~Chenu, J.~Molina-Vilaplana and A.~del Campo, \emph{Work statistics, Loschmidt echo
and information scrambling in chaotic quantum systems}, Quantum \textbf{3}, (2019) 127
  




\bibitem{CES21}
G. Cipolloni, L. Erd\"os and D. Schr\"oder, \emph{On the spectral form factor for random matrices},
arXiv:2109.06712.
  
\bibitem{Co21} 
   P. Cohen, \emph{Moments of discrete classical
orthogonal polynomial ensembles}, arXiv:2112.02064.
  
\bibitem{CCO20}  
  P. Cohen, F.D. Cunden and N. O'Connell,  \emph{Moments of discrete orthogonal polynomial ensembles},
Electron. J. Probab.  \textbf{25} (2020), 1--19.

 \bibitem{Co05}
B.~Collins, \emph{Product of random projections, {J}acobi ensembles and
  universality problems arising from free probability}, Prob. Theory Rel.
  Fields \textbf{133} (2005), 315--344.


\bibitem{C+17} J.S. Cotler, G. Gur-Ari, M. Hanada, J. Polchinski, P. Saad, S.H. Shenker, D. Stanford,
A. Streicher and M. Tezuka, \emph{Black Holes and Random Matrices}, JHEP \textbf{1705} (2017), 118;
Erratum: [JHEP \textbf{1809} (2018), 002]



\bibitem{CH19} J.S. Cotler and N.~Hunter-Jones, \emph{Spectral decoupling in many-body quantum
chaos},  JHEP \textbf{12}
(2020) 205

\bibitem{CHLY17}
J.S. Cotler, N. Hunter-Jones, J. Liu and B. Yoshida, \emph{Chaos, Complexity, and Random
Matrices}  JHEP \textbf{1711} (2017), 048 
  
  
\bibitem{CMOS19}
F. Cunden, F. Mezzadri, N. O'Connell and N. Simm,
 \emph{Moments of random matrices and hypergeometric orthogonal polynomials}, Comm. Math. Phys. \textbf{369}  (2019), 1091--1145.  
 
 
\bibitem{De10} 
N.~ Demni,  \emph{$\beta$-Jacobi processes} Adv. Pure Appl. Math. \textbf{1} (2010), 325--344.
 
 \bibitem{DE05}
I.~Dumitriu and A.~Edelman, \emph{Global spectrum fluctuations for the $\beta$-{H}ermite and
  $\beta$-{L}aguerre ensembles via matrix models}, J. Math. Phys. \textbf{47}
  (2006), 063302.
  
   \bibitem{DP12}
I.~Dumitriu and E.~Paquette, \emph{Global fluctuations for linear statistics of
  $\beta$ {J}acobi ensembles}, Random Matrices: Theory Appl. \textbf{01}
  (2012), 1250013.


\bibitem{DO81}
B. Durhuus and P. Olesen, \emph{The spectral density for two-dimensional continuum QCD}, Nucl.
Phys. B 184 (1981) 461--475.



\bibitem{Dy62b}
F.J. Dyson, \emph{A {B}rownian motion model for the eigenvalues of a random
  matrix}, J. Math. Phys. \textbf{3} (1962), 1191--1198.
  
  
  \bibitem{Er55}
  A. Erd\'elyi,  \emph{Asymptotic representations of Fourier integrals and the method of stationary
phase}, J. Soc. Indust. Appl. Math., \textbf{3} (1955), 17--27.
  
\bibitem{EKS20}  
  L. Erd\"os, T. Kr\"uger, and D. Schr\"oder, \emph{Cusp universality for random matrices I: Local
law and the complex Hermitian case}, Comm. Math. Phys., \textbf{378} (2020), 1203--1278.
  
  \bibitem{EY17}
L. Erd\"os and H.-T. Yau, \emph{A dynamical approach to random matrix theory},
Courant Lecture Notes in Mathematics, vol. 28, Amer. Math. Soc. Providence, 2017.

 \bibitem{FFN10}
B.~Fleming, P.J.~Forrester and E.~Nordenstam,
\emph{A finitization of the bead process},
Probab.  Th. Related Fields (2012)   \textbf{152}, 321--356.



\bibitem{Fo90b}
P.J.~Forrester, \emph{Exact solution of the lockstep model of vicious walkers}, J.
  Phys. A \textbf{23} (1990), 1259--1273.
  
\bibitem{Fo96}  
  P.J. Forrester, \emph{Some exact correlations in the Dyson Brownian motion model for transitions
  to the CUE},
  Physica A  \textbf{223} (1996), 365--390.
  
    \bibitem{Fo10}
P.J. Forrester, \emph{Log-gases and random matrices}, Princeton University Press,
  Princeton, NJ, 2010.
  
     \bibitem{Fo18}
P.J.~Forrester,  \emph{Meet Andr\'eief, Bordeaux 1886, and Andreev, Kharkov 1882--83},
Random Matrices Theory Appl.  
 \textbf{8} (2019) 1930001.  
 
 \bibitem{Fo21a} 
 P.J. Forrester, \emph{Differential identities for the structure function of some random matrix ensembles},
J. Stat. Phys. \textbf{183} (2021), 28.

 \bibitem{Fo21b}
 P. J. Forrester, \emph{Quantifying dip-ramp-plateau for the Laguerre unitary ensemble structure function}, 
 Commun. Math. Phys. \textbf{387} (2021), 215--235.
 
 
  \bibitem{Fo22}
P.J.~Forrester,  \emph{Global and local scaling limits for the $\beta = 2$
Stieltjes--Wigert random matrix ensemble},
Random Matrices Theory Appl. \textbf{11} (2022), 2250020.

  \bibitem{Fo22a}
P.J.~Forrester,  \emph{Joint moments of a characteristic polynomial and its derivative for
the circular $\beta$--ensemble}, Probab. Math. Phys. \textbf{3} (2022), 145--170.

 \bibitem{Fo22b}
P.J.~Forrester,  \emph{High--low temperature dualities for the classical $\beta$-ensembles},
Random Matrices Theory Appl. (2022) https://doi.org/10.1142/S2010326322500356

 \bibitem{Fo22c}
P.J.~Forrester,  \emph{A review of exact results for fluctuation formulas  in random matrix
theory}, arXiv:2204.03303.

  
 \bibitem{FG16}  
  P.J. Forrester and J. Grela, \emph{Hydrodynamical spectral evolution for random matrices}, J.
Phys. A,  \textbf{49} (2016), 085203.

 \bibitem{FK22} P.J. Forrester and S.~Kumar, \emph{Differential recurrences for the distribution of the trace of the
$\beta$-Jacobi ensemble}, Physica D \textbf{434} (2022), 133220.

 \bibitem{FLSY21}
P.J. Forrester, Shi-Hao Li, Bo-Jian Shen and Guo-Fu Yu, \emph{$q$-Pearson pair and moments in $q$-deformed ensembles},
arxiv:2110.13420.

\bibitem{FR21}
P. Forrester and A. Rahman, \emph{Relations between moments for the Jacobi and Cauchy random matrix ensembles}, J. Math. Phys. 62 (2021), 073302.

      \bibitem{FRW17}
P.J. Forrester, A.A. Rahman, and N.S. Witte, \emph{Large $N$ expansions for the Laguerre and Jacobi $\beta$ ensembles from the loop equations}, J. Math. Phys. \textbf{58}
  (2017), 113303.
  
    \bibitem{FW17} 
  P.J. Forrester and D. Wang,  \emph{Muttalib-Borodin ensembles in random matrix theory --- realisations and correlation functions}, { Elec. J. Probab.} \textbf{22} (2017), 54.

 \bibitem{FLD16}
Y.V.~Fyodorov and P.~Le~Doussal, \emph{Moments of the position of the maximum for GUE characteristic polynomials and for log-correlated Gaussian processes}, J. Stat. Phys.\textbf{164}
  (2016), 190--240.
  


\bibitem{GT14}
G. Giasemidis and M. Tierz, \emph{Torus knot polynomials and SUSY Wilson loops}, Lett. Math.
Phys. \textbf{104} (2014), 1535--1556.  

\bibitem{GS15}
V. Gorin and M. Shkolnikov, \emph{Multilevel Dyson Brownian motions via Jack polynomials}, Probab. Th.
Related Fields, \textbf{163} (2015), 413--463.

\bibitem{GM95}  
  D.J. Gross and A. Matytsin, \emph{Some properties of large N two-dimensional Yang-Mills theory},
Nucl. Phys. B \textbf{437} (1995), 541--584.

 \bibitem{HHN16}
  W.~Hachem, A.~Hardy and J.~Najim, \emph{Large complex correlated Wishart matrices: the
  Pearcey kernel and expansion at the hard edge}, Elec. J. Probab. \textbf{21} (2016), 1--36.

  
 \bibitem{Ha01}
B.C.~Hall, \emph{Harmonic analysis with respect to heat kernel measure},  Bull. Amer. Math. Soc.  \textbf{38}
(2001), 43--78.

 \bibitem{Ha18}
 T.~Hamdi, \emph{Spectral distribution of the free Jacobi process revisited},
  Anal. PDE. \textbf{11} (2018),  2137--2148.
  
 \bibitem{HL19}  
  J. Huang and B. Landon, \emph{Rigidity and a mesoscopic central limit theorem for Dyson Brownian motion for
general $\beta$ and potentials}, Probab. Th. Related Fields, \textbf{175} (2019), 209--253.
  

 \bibitem{Hu56}
G.A. Hunt,  \emph{Semigroups of measures on Lie groups}, Trans. Am. Math. Soc. \textbf{81} (1956),
264--293.

 \bibitem{It50}
K. It\^{o}, \emph{Brownian motions in a Lie group}, Proc. Jap. Acad. 26 (1950), 4--10.

\bibitem{Ka22}
Z.~Kabluchko, \emph{Lee-Yang zeros of the Curie-Weiss ferromagnet, unitary Hermite polynomials,
and the backward heat flow}, arXiv:2203.05533.


  
   \bibitem{Ka16}
M.~Katori, \emph{Bessel processes, {S}chramm--{L}oewner evolution, and the {D}yson
  model}, Springer briefs in mathematical physics, vol.~11, Springer, Berlin,
  2016.
  
  
 \bibitem{Ka81}  
  V. A. Kazakov,  \emph{Wilson loop average for an arbitrary contour in two-dimensional u(n) gauge theory}, Nucl. Phys. B \textbf{179} (1981), 283--292.
  
  
  \bibitem{Ke17} 
  T. Kemp, \emph{Heat kernel empirical laws on U(N) and GL(N)},  J. Theoret. Probab., \textbf{30} (2017), 397--451.
  
   \bibitem{KLZF20}  
  M. Kieburg, S.-H. Li, J. Zhang, and P. J. Forrester, \emph{Cyclic P\'olya ensembles on the unitary matrices and their
spectral statistics},  arXiv:2012.11993.




\bibitem{LSY19}
B. Landon, P. Sosoe, and H.-T. Yau. \emph{Fixed energy universality of Dyson Brownian motion}, Adv.
Math., \textbf{346} (2019), 1137--1332.

\bibitem{LLJP86} 
 L. Leviandier, M. Lombardi, R. Jost and J. P. Pique,
 \emph{Fourier transform: a tool to measure
statistical level properties in very complex spectra}, Phys. Rev. Lett. \textbf{56} (1986), 2449.


  
  
 \bibitem{Le08}  
 T.~L\'evy, \emph{Schur-Weyl duality and the heat kernel measure on the unitary group}, Adv. Math. \textbf{218}, (2008)
537--575.

\bibitem{LM10}
T. L\'evy and M. Ma\"ida, \emph{Central limit theorem for the heat kernel measure on the unitary group}, J. Funct. Anal.,
\textbf{259} (2010), 3163--3204.

 \bibitem{LPC21}
 J. Li, T. Prosen, and A. Chan, 
 \emph{Spectral statistics of non-Hermitian matrices and dissipative quantum chaos},
Phys. Rev. Lett. \textbf{127} (2021), 170602.

 \bibitem{LW16}
K. Liechty, and D. Wang, \emph{Nonintersecting Browninan motion on the unit circle},  Ann.
Prob.  \textbf{44}, (2016), 1134--1211.


\bibitem{Li58}
J. Lighthill,
\emph{Introduction to Fourier analysis and generalized functions},
Cambridge University Press, {1958}.



\bibitem{Ma95}
I.G. Macdonald, \emph{Hall polynomials and symmetric functions}, 2nd ed., Oxford
  University Press, Oxford, 1995.
  
  
  \bibitem{Ma21}
  A. W. Marcus, \emph{Polynomial convolutions and (finite) free probability}, arXiv:2108.07054.
  
  \bibitem{MRW15}
F.~Mezzadri, A.K. Reynolds and B.~Winn, \emph{Moments of the eigenvalue densities and of the secular coefficients of $\beta$-ensembles}, Nonlinearity \textbf{30}
  (2017), 1034.
  
   \bibitem{Me04}
M.L. Mehta, \emph{Random matrices}, 3rd ed., Elsevier, San Diego, 2004.

 \bibitem{Mi21}  
B. Mirabelli, \emph{Hermitian, non-Hermitian and multivariate finite free probability},
 PhD Thesis, Princeton, 2021
University.  
  
  
  
   \bibitem{MMPS12}  
  A.D. Mironov, A.Yu. Morozov, A.V. Popolitov, and Sh.R.Shakirov, \emph{Resolvents and Seiberg-Witten representation for a Gaussian $\beta$-ensemble}, Theor. Math. Phys., \textbf{171} (2012), 505--522.
  
 \bibitem{MH21}
 A.~Mukherjee and S.~Hikami, \emph{Spectral form factor for time-dependent matrix model}, JHEP \textbf{2021} (2021) 071. 
  
  
 \bibitem{Ok19}
  K. Okuyama, \emph{Spectral form factor and semi-circle law in the time direction},
  JHEP \textbf{2019} (2019), 161.   
  
  \bibitem{On81}
     E. Onofri,  \emph{SU(N) Lattice gauge theory with Villain's action}, Nuovo Cim. A  \textbf{66} (1981), 293--318.

  
  \bibitem{PS91}
A. Pandey and P. Shukla,  \emph{Eigenvalue correlations in the circular ensembles}, J. Phys. A  \textbf{24} (1991),
3907--3926.

 \bibitem{Ra97}
 E.M.~Rains, \emph{Combinatorial properties of Brownian motion on the compact classical groups}, J. Theoret.
Probab. \textbf{10} (1997), 659--679.


 \bibitem{Ro81}
P. Rossi,  \emph{Continuum QCD${}_2$ from a fixed point lattice action}, Ann. Phys. \textbf{132} (1981), 463--481.


\bibitem{Sh16} S.~Shenker, \emph{Black holes and random matrices},  [web resource from the conference Natifest, 2016].

  \bibitem{Su71a}
B.~Sutherland, \emph{Exact results for a quantum many body problem in one
  dimension}, Phys. Rev. A \textbf{4} (1971), 2019--2021.



\bibitem{SZ22}
O. Szehr and R. Zarouf,
\emph{On the asymptotic behavior of Jacobi polynomials with first varying parameter},
J. Approx.  Th.
\textbf{277}
(2022), 105702.
  
  
  \bibitem{TGS18}
E.J.~Torres-Herrera, A.M.~Garc\'ia-Garc\'ia,  and L.F.~Santos, 
\emph{Generic dynamical features of quenched interacting quantum systems: 
Survival probability, density  imbalance,  and  out-of-time-ordered  correlator}, Phys. Rev. B
\textbf{97} (2018), 060303.


 \bibitem{Ul85} N.~Ullah, \emph{Probability density function of the single eigenvalue
  outside the semicircle using the exact Fourier transform},
  J.~Math.~Phys. \textbf{26}, (1985), 2350--2351.
  
 \bibitem{VP09}
 Vinayak and A.~Pandey,  \emph{Transition from Poisson to circular unitary ensemble}, Pramana - J Phys \textbf{73} (2009), 505?519.


\bibitem{VG21} W.L. Vleeshouwers and V.~Gritsev, \emph{Topological field theory approach to intermediate
statistics}, SciPost Phys. \textbf{10} (2021), 146.


\bibitem{VG22} W.L. Vleeshouwers and V.~Gritsev,
\emph{The spectral form factor in the 't Hooft limit --- intermediacy versus universality}, 
SciPost Phys. Core \textbf{5} (2022), 051. 


 \bibitem{We16}
C. Webb, \emph{Linear statistics of the circular $\beta$-ensemble, Stein's method, and circular
Dyson Brownian motion}, Electron. J. Probab. \textbf{21} (2016), 25.

  \bibitem{WW65}
E.T.~Whittaker and G.N.~Watson, \emph{A course of modern analysis}, 4th ed.,
  Cambridge University Press, Cambridge, 1927.
  
    \bibitem{WF14}
N.S. Witte and P.J. Forrester, \emph{Moments of the {G}aussian $\beta$ ensembles
  and the large {$N$} expansion of the densities}, J. Math. Phys. \textbf{55}
  (2014), 083302.
  
   \bibitem{WF15}
N.S. Witte and P.J. Forrester, \emph{Loop equation analysis of the circular ensembles}, JHEP
  \textbf{2015} (2015), 173.
  
  \bibitem{Wo11}
  Wolfram Mathematica, online support, www.wolfram.com/language/11/random-matrices/dyson-coulomb-gas.html?product=mathematica
  
    \bibitem{Ya20} C.~Yan, \emph{Spectral form factor}, [web resource dated June 28, 2020].
  
  
  \bibitem{Yo52}  
  K. Yosida,  \emph{Brownian motion in a homogeneous Riemannian space}, Pac. J. Math. \textbf{2} (1952),
263--270.

\bibitem{ZKF21}
J.~Zhang, M.~Kieburg and P.J.~Forrester,
\emph{Harmonic analysis for rank-1 randomised Horn problems}, Lett. Math. Phys. \textbf{111} (2021), 98.
  
  
\end{thebibliography}
\end{document}